\documentclass[11pt,a4paper]{article}
\pdfoutput=1

\makeatletter
\providecommand{\@listi}{\leftmargin\leftmargini}
\makeatother

\usepackage{jheppub}

\usepackage{amsthm}
\usepackage{booktabs}
\usepackage{array}

\makeatletter\toccontinuoustrue\makeatother

\makeatletter
\def\@listI{\leftmargin\leftmargini \parsep 2\p@ \topsep 3\p@ \itemsep 2\p@}
\let\@listi\@listI
\makeatother

\makeatletter
\def\th@plain{\thm@headfont{\bfseries}\itshape\thm@preskip6\p@\thm@postskip6\p@}
\makeatother

\newtheorem{theorem}{Theorem}[section]
\newtheorem{lemma}[theorem]{Lemma}
\newtheorem{proposition}[theorem]{Proposition}
\newtheorem{corollary}[theorem]{Corollary}

\newcommand{\ZZ}{\mathbb{Z}}
\newcommand{\RR}{\mathbb{R}}
\newcommand{\rep}[1]{\mathbf{#1}}
\newcommand{\arep}[1]{\overline{\mathbf{#1}}}
\newcommand{\Tr}{\mathrm{Tr}}
\newcommand{\Wsm}{W_{\rm SM}}
\newcommand{\Nosc}{N_{\rm osc}}
\newcommand{\ten}{$\rep{10}$-like}
\newcommand{\fbar}{$\arep{5}$-like}

\title{A localisation bound for Standard-Model representations\\ in heterotic orbifolds}

\author{Navid Ardakanian}
\affiliation{Independent Researcher}
\emailAdd{n.ardakanian@gmail.com}

\abstract{In heterotic string compactifications on toroidal orbifolds, the massless states at a twisted fixed point are constrained by the evenness of the sixteen-dimensional gauge lattice. We show that two massless left-chiral states in the same Standard-Model representation, built on the same constructing space-group element and distinct modulo the Weyl group of $SU(3)_c\times SU(2)_L$, must satisfy $\sqrt{R}+\sqrt{R'}\ge\sqrt2$, where $R$ is the squared length of the part of a state's gauge momentum orthogonal to the Standard-Model Cartan directions. For states without left-moving oscillator excitations this reads $R\ge1/2$. When $R<1$, $n$ such states, pairwise distinct in the same sense, carry the weights of the fundamental representation of a local $SU(n)$ commuting with the Standard Model, and $n\le1/(1-R)$. Because the representation fixes the projected length --- at the canonical hypercharge normalisation, $6/5$ for the $\mathbf{10}$-like $Q$, $u^c$, $e^c$ and $4/5$ for the $\bar{\mathbf{5}}$-like $d^c$ and lepton doublet --- the bound becomes a condition on the sector: a repeated $\mathbf{10}$-like representation requires a massless budget of at least $17/20$, and a repeated $\bar{\mathbf{5}}$-like one at least $13/20$. In the complete massless spectra of $16{,}138$ models from two orbifold catalogues the multiplicity ceiling is reached at four values of $R$. We exhibit an explicit model on $T^6/(\ZZ_3\times\ZZ_6)$ whose budget clears the higher threshold and carries a repeated $\mathbf{10}$-like representation at the multiplicity the bound permits. For the Standard Model embedded through $SU(5)$ in one $E_8$ factor we show that, for that geometry and budget, exactly two inequivalent order-six shift classes admit such a doubling, and among the embeddings over them with the Standard Model intact we find none with three net generations of $Q$, $u^c$ and $e^c$ at once.}


\begin{document}

\maketitle

%%%%%%%%%%%%%%%%%%%%%%%%%%%%%%%%%%%%%%%%%%%%%%%%%%%%%%%%%%%%%%%%%%%%%%%%%%%%%%%
\section{Introduction}
\label{sec:intro}

Toroidal orbifolds of the $E_8\times E_8$ heterotic string \cite{Gross:1984dd,Gross:1985fr,Dixon:1985jw,Dixon:1986jc,Ibanez:1986tp} are among the most thoroughly explored routes from string theory to the supersymmetric Standard Model \cite{Buchmuller:2005jr,Buchmuller:2006ik,Lebedev:2006kn,Nilles:2008gq}. Much of their phenomenology follows from where the matter sits. Untwisted states propagate through the whole compact space, twisted states are localised at the fixed points of the orbifold action, and at a fixed point the gauge shift and the Wilson lines leave a local gauge group that is typically larger than the four-dimensional one \cite{Forste:2004ie}. The local group can be grand-unified even when the four-dimensional group is that of the Standard Model, and the matter localised there then fills complete multiplets of it; a $\rep{16}$ of a local $SO(10)$ carrying one generation is the standard example \cite{Kobayashi:2004ya,Buchmuller:2006ik,Lebedev:2006kn}. Equivalent fixed points carry copies of the same local multiplet, which is one way families are replicated \cite{Kobayashi:2004ya,Buchmuller:2005jr}, and the permutations of equivalent fixed points, together with the space-group selection rules, act on the localised families as non-Abelian discrete flavour symmetries such as $D_4$ and $\Delta(54)$ \cite{Kobayashi:2004ya,Kobayashi:2006wq}. The couplings follow the geometry as well. A Yukawa coupling can come directly from the gauge interactions and be of the order of the gauge coupling, as the top-quark coupling does in the model of refs.~\cite{Buchmuller:2005jr,Buchmuller:2006ik}; couplings among twisted states at different fixed points are mediated by worldsheet instantons \cite{Hamidi:1986vh,Dixon:1986qv} and depend strongly on the relative position of the fields \cite{Forste:2004ie}; and which couplings are allowed at all is decided by string selection rules \cite{Kobayashi:2006wq,CaboBizet:2013gns}.

It is therefore natural to ask how much matter of one kind a single fixed point can hold, and the answer is arithmetic. A twisted state is built on one element of the space group, and that element fixes both the locus on which the state sits --- a point in the directions the twist rotates --- and the gauge shift it feels. Two massless states with the same Standard-Model quantum numbers, built on the same element and not related by a Weyl reflection of $SU(3)_c\times SU(2)_L$, have gauge momenta whose difference is a non-zero vector of the gauge lattice. The lattice is even, so the difference has squared length at least two, and the triangle inequality turns this into a lower bound on the massless budget of the sector. Under the hypotheses of section~\ref{sec:bound}, which require the Standard-Model factors to survive at level one and the two states to share their constructing element and their representation, evenness is the only property of the lattice that enters: the local gauge group, the geometry beyond the vacuum-energy shift, the Wilson lines and the number of $U(1)$ factors all drop out.

Applied to several states at once, the same inequality fixes their structure. Let $R$ be the squared length of the part of a state's gauge momentum orthogonal to the Standard-Model Cartan directions. When $R<1$, the pairwise differences of $n$ oscillator-free states of one representation at one fixed point are roots spanning an $A_{n-1}$ system orthogonal to the Standard Model, and the states carry the weights of the fundamental representation of the $SU(n)$ these roots generate inside the local gauge group. The bound on $n$ becomes the statement that the level-one conformal weight of that representation, $(n-1)/2n$, cannot exceed $R/2$. Since a Standard-Model representation fixes the length of the projection of its momentum onto the Cartan directions, and at the canonical normalisation of hypercharge that length takes only the two values of the $\rep{10}$ and the $\rep5$ of $SU(5)$, the condition becomes a threshold on the massless budget of a sector, higher for $Q$, $u^c$ and $e^c$ than for $d^c$ and the lepton doublet.

Section~\ref{sec:setup} fixes the conventions and section~\ref{sec:norms} derives the two projected norms. Section~\ref{sec:bound} proves the bound, and section~\ref{sec:ceiling} the local $SU(n)$ and the multiplicity ceiling that follows from it. Section~\ref{sec:verif} surveys how closely the complete massless spectra of two orbifold catalogues approach that ceiling. Section~\ref{sec:construction} exhibits a model in which a \ten{} representation is doubled at a fixed point and classifies, for its geometry and budget, the shift vectors that allow this; section~\ref{sec:signature} shows how such a configuration appears in a spectrum on that geometry, and section~\ref{sec:discussion} discusses the scope of the results.

%%%%%%%%%%%%%%%%%%%%%%%%%%%%%%%%%%%%%%%%%%%%%%%%%%%%%%%%%%%%%%%%%%%%%%%%%%%%%%%
\section{Gauge momenta at a twisted fixed point}
\label{sec:setup}

We consider the $E_8\times E_8$ heterotic string \cite{Gross:1984dd,Gross:1985fr} on a toroidal orbifold $T^6/(\ZZ_M\times\ZZ_N)$ \cite{Dixon:1985jw,Dixon:1986jc} with point-group generators $\theta,\omega$, gauge shifts $V_1,V_2$ and discrete Wilson lines $W_a$, $a=1,\dots,6$ \cite{Ibanez:1986tp}. The gauge degrees of freedom live on $\Lambda=\Lambda_{E_8}\oplus\Lambda_{E_8}\subset\RR^{16}$, which is even, $|\lambda|^2\in2\ZZ$ for every $\lambda\in\Lambda$, and therefore integral, since $\lambda\cdot\mu=\tfrac12(|\lambda+\mu|^2-|\lambda|^2-|\mu|^2)$. Its roots, the vectors of squared length $2$, number $240$ in each factor.

A twisted sector is constructed by a space-group element $g=(\theta^{k}\omega^{l},\sum_a n_ae_a)$ with $k,l,n_a\in\ZZ$. Its point-group part acts on the gauge degrees of freedom through the shifts and its translation through the Wilson lines, so that the \emph{local shift} at $g$ is
\begin{equation}
\label{eq:Vg}
V_g \;=\; k\,V_1+l\,V_2+\sum_a n_a W_a .
\end{equation}
A massless left-chiral state in the sector carries a shifted gauge momentum $p\in V_g+\Lambda$ of fixed length,
\begin{equation}
\label{eq:budget}
\tfrac12|p|^2 \;=\; 1-\delta c(g)-\Nosc \;=:\; B ,
\end{equation}
where $\delta c(g)=\tfrac12\sum_{a=1}^{3}\eta_a(1-\eta_a)$ is the vacuum-energy shift built from the twist eigenvalues $\eta_a\in[0,1)$ of the point-group part of $g$ and $\Nosc\ge0$ is the left-moving oscillator number. We call $B$ the \emph{massless budget} and take \eqref{eq:budget} in the conventions of the \textsc{orbifolder} \cite{Nilles:2011aj}, which we use throughout. The admissible oscillator numbers are sums of the $\eta_a$ and $1-\eta_a$ over the fractionally twisted planes, so the smallest excitation on a geometry can be a small fraction. Because $\delta c$ depends on the point-group part of $g$ alone, and neither on its translation, nor on the Wilson lines, nor on the state, two states built on one element satisfy
\begin{equation}
\label{eq:budgetdiff}
B-B'=\Nosc'-\Nosc ,
\end{equation}
and two oscillator-free states in particular have equal budgets.

Let $W\subset\RR^{16}$ be the span of the two $SU(3)_c$ Cartan directions, the $SU(2)_L$ Cartan direction and the hypercharge generator $t_Y$, so that $\dim W=4$. Here $t_Y$ is the hypercharge generator of the model, as given by the \textsc{orbifolder}. With $P_W$ the orthogonal projector onto $W$ we write
\begin{equation}
\label{eq:split}
p=w+s,\qquad w=P_Wp,\quad s=p-w\perp W,\qquad R:=|s|^2=2B-|w|^2 ,
\end{equation}
and, for any vector $x$, $x_\perp:=x-P_Wx$, so that $s=p_\perp$. The \emph{Standard-Model representation} of a state is its set of quantum numbers under $SU(3)_c\times SU(2)_L\times U(1)_Y$ when these are those of $Q$, $u^c$, $e^c$, $d^c$, $L$ or the doublet of hypercharge $+\tfrac12$, or of the conjugate of one of them. When two momenta share their $W$-projection we write $\Delta$ for their difference; since the pair is unordered, $\Delta$ is defined up to sign, and its pairing with a momentum appears as $|p\cdot\Delta|$.

Two states sit at the \emph{same twisted fixed point} when one space-group element, point-group part and translation alike, constructs both, so that the same combination of gauge shifts and Wilson lines acts on both and the same vacuum-energy shift applies. Conjugate elements may be used interchangeably: for $g=(\theta_g,\lambda)$ and $h=(\theta_h,\lambda_h)$ one has $hgh^{-1}=(\theta_g,(I-\theta_g)\lambda_h+\theta_h\lambda)$, so $V_{hgh^{-1}}-V_g$ is a sum of two terms that vanish modulo $\Lambda$ by the Wilson-line consistency conditions, while $\delta c$ depends only on the point-group part, which conjugation leaves fixed. The spectrum lists one element per conjugacy class, and we use \emph{fixed point} and \emph{constructing element} interchangeably in this sense. Identifying fixed points by their local shift would manufacture repetitions, since distinct fixed points of one sector often share it: in the first of the two catalogues of section~\ref{sec:verif}, $674{,}558$ pairs of distinct fixed points of one sector that carry massless states in the same model do. A point of the covering torus errs the other way, since one listed element is a class of space-group elements which, except where the twist leaves a single fixed point per sector, contains several such points.

%%%%%%%%%%%%%%%%%%%%%%%%%%%%%%%%%%%%%%%%%%%%%%%%%%%%%%%%%%%%%%%%%%%%%%%%%%%%%%%
\section{Projected norms of Standard-Model representations}
\label{sec:norms}

At the canonical normalisation of hypercharge each Standard-Model representation fixes $|w|^2$, and only two values occur.

\begin{proposition}[projected norms]
\label{prop:norms}
Let $SU(3)_c\times SU(2)_L\times U(1)_Y$ be identified in the model, with hypercharge at the level $k_Y:=2|t_Y|^2=5/3$. Then
\begin{equation}
\label{eq:wnorms}
|w|^2=\tfrac65 \ \text{ for } Q,\ u^c,\ e^c ,
\qquad
|w|^2=\tfrac45 \ \text{ for } d^c \text{ and the lepton doublet},
\end{equation}
and conjugate representations have the same $|w|^2$.
\end{proposition}

\begin{proof}
The span $W$ splits orthogonally as $W=W_3\oplus W_2\oplus\RR t_Y$: the $SU(3)_c$ and $SU(2)_L$ Cartans are orthogonal because the two factors commute and their roots are mutually orthogonal, and $t_Y$ is orthogonal to every root of either factor because $U(1)_Y$ commutes with both. Hence $|w|^2=|w_3|^2+|w_2|^2+|w_Y|^2$. With roots normalised to $|\alpha|^2=2$, the fundamental weight of $A_{n-1}$ has $|\omega_1|^2=(n-1)/n$, so $|w_3|^2=2/3$ for a $\rep3$ or $\arep3$ and $0$ for a singlet; the doublet weights are $\pm\alpha/2$, so $|w_2|^2=|\alpha|^2/4=1/2$ for a doublet and $0$ for a singlet; and $w_Y=\bigl((p\cdot t_Y)/|t_Y|^2\bigr)t_Y$ gives $|w_Y|^2=2Y^2/k_Y$, with $Y=p\cdot t_Y$ the hypercharge in the convention in which $t_Y$ is written, which is $\tfrac65Y^2$ at $k_Y=5/3$. For $Q$ and $L$ this reads
\begin{equation*}
|w_Q|^2=\tfrac23+\tfrac12+\frac{(1/6)^2}{5/6}=\tfrac{20}{30}+\tfrac{15}{30}+\tfrac1{30}=\tfrac65 ,
\qquad
|w_L|^2=\tfrac12+\frac{(1/2)^2}{5/6}=\tfrac12+\tfrac3{10}=\tfrac45 ,
\end{equation*}
and table~\ref{tab:norms} lists all six cases. The weights of a conjugate representation are the negatives of those of the representation, which leaves $|w|^2$ unchanged.
\end{proof}

\begin{table}[t]
\centering
\begin{tabular}{lccccc}
\toprule
representation & $(SU(3)_c,SU(2)_L,Y)$ & $|w_3|^2$ & $|w_2|^2$ & $\tfrac65Y^2$ & $|w|^2$\\
\midrule
$Q$   & $(\rep3,\rep2,+\tfrac16)$        & $2/3$ & $1/2$ & $1/30$  & $6/5$ \\
$u^c$ & $(\arep3,\rep1,-\tfrac23)$       & $2/3$ & $0$   & $8/15$  & $6/5$ \\
$e^c$ & $(\rep1,\rep1,+1)$               & $0$   & $0$   & $6/5$   & $6/5$ \\
$d^c$ & $(\arep3,\rep1,+\tfrac13)$       & $2/3$ & $0$   & $2/15$  & $4/5$ \\
$L$   & $(\rep1,\rep2,-\tfrac12)$        & $0$   & $1/2$ & $3/10$  & $4/5$ \\
$\bar L$ & $(\rep1,\rep2,+\tfrac12)$     & $0$   & $1/2$ & $3/10$  & $4/5$ \\
\bottomrule
\end{tabular}
\caption{Projected norms of the six Standard-Model representations at $k_Y=5/3$.}
\label{tab:norms}
\end{table}

The two values are those of $SU(5)$. At level one a primary of weight $\lambda$ of a simply-laced algebra has conformal weight $h=|\lambda|^2/2$, so $|w|^2=2h$. With its induced inner product, $W=W_3\oplus W_2\oplus\RR t_Y$ is isometric to the Cartan subspace of $SU(5)$, the $SU(3)_c$ and $SU(2)_L$ Cartans and $t_Y$ going to those of the standard embedding of $SU(3)\times SU(2)\times U(1)_Y$, because at $k_Y=5/3$ the length $|t_Y|^2=5/6$ is that of $\mathrm{diag}(-\tfrac13,-\tfrac13,-\tfrac13,\tfrac12,\tfrac12)$. For $A_4$ with $|\alpha|^2=2$ one has $|\omega_i|^2=i(5-i)/5$, hence $h(\rep{10})=3/5$ and $h(\rep5)=2/5$, which reproduces $6/5$ and $4/5$; for $Q$ the three contributions are $\tfrac13+\tfrac14+\tfrac1{60}=\tfrac35$. Only the Cartan subspace enters, and no unbroken $SU(5)$ is assumed. We call $Q$, $u^c$, $e^c$ and their conjugates the \emph{\ten} representations and $d^c$, $L$ and their conjugates the \emph{\fbar} ones, after the $SU(5)$ multiplets whose norm they share. Under $SU(5)$, $\rep5=(\rep3,\rep1)_{-1/3}\oplus(\rep1,\rep2)_{+1/2}$ and $\arep5=(\arep3,\rep1)_{+1/3}\oplus(\rep1,\rep2)_{-1/2}=d^c\oplus L$, so $H_u$ sits in a $\rep5$ and $H_d$ in a $\arep5$. The doublet $(\rep1,\rep2)_{+1/2}$ that phenomenology calls $H_u$ is the conjugate of $L$, so up to conjugation there are five Standard-Model representations.

The normalisation of hypercharge deserves a comment. A projector onto a span is unchanged when a basis vector is rescaled, so $|w|^2$ does not depend on the normalisation of $t_Y$; the numerals $6/5$ and $4/5$ do. With the Standard-Model hypercharges fixed in the convention in which $t_Y$ is written, $|w|^2=|w_3|^2+|w_2|^2+2Y^2/k_Y$, and
\begin{gather*}
|w_Q|^2=|w_{e^c}|^2\iff\tfrac76+\tfrac1{18k_Y}=\tfrac2{k_Y}\iff k_Y=\tfrac53,\\
|w_{d^c}|^2=|w_{L}|^2\iff\tfrac23+\tfrac2{9k_Y}=\tfrac12+\tfrac1{2k_Y}\iff k_Y=\tfrac53 ,
\end{gather*}
and likewise $|w_{u^c}|^2=|w_{e^c}|^2\iff k_Y=\tfrac53$, each equality having this as its only positive solution. The norms are equal within each class exactly at the grand-unified normalisation, which in heterotic orbifold models is what an embedding of the Standard Model into a simple local group containing $SU(5)$ provides \cite{Lebedev:2006kn}. We take the hypercharge generator from the model, as we do the $SU(3)_c$ and $SU(2)_L$ factors. In every model of the two catalogues of section~\ref{sec:verif} it has $k_Y=5/3$, which is a useful check on the identification of the Standard Model.

%%%%%%%%%%%%%%%%%%%%%%%%%%%%%%%%%%%%%%%%%%%%%%%%%%%%%%%%%%%%%%%%%%%%%%%%%%%%%%%
\section{An even-lattice bound}
\label{sec:bound}

\begin{theorem}[localisation bound]
\label{thm:bound}
Let $p$ and $p'$ be the gauge momenta of two massless left-chiral states in a twisted sector, with budgets $B,B'$. Suppose:
\begin{enumerate}
\item[\textbf{H1}] the gauge lattice $\Lambda=\Lambda_{E_8}\oplus\Lambda_{E_8}$ is even;
\item[\textbf{H2}] $SU(3)_c\times SU(2)_L\times U(1)_Y$ is contained in the four-dimensional gauge group; equivalently, every $SU(3)_c$ and $SU(2)_L$ root $\alpha$ --- an element of $\Lambda$ with $|\alpha|^2=2$, the Standard-Model factors being realised at level one --- satisfies $\alpha\cdot V_1,\ \alpha\cdot V_2,\ \alpha\cdot W_a\in\ZZ$ for every $a$;
\item[\textbf{H3}] both states are built on the same constructing space-group element $g$ in the sense of section~\ref{sec:setup}, hence on the same local shift $V_g$ of \eqref{eq:Vg}, Wilson lines included, and the same $\delta c(g)$;
\item[\textbf{H4}] the two states carry the same Standard-Model representation, hence the same $|w|^2$;
\item[\textbf{H5}] there is no $\sigma\in\Wsm$ with $p'=\sigma(p)$, where $\Wsm$ is the group generated by the reflections $\sigma_\alpha(x)=x-(x\cdot\alpha)\alpha$ in the $SU(3)_c$ and $SU(2)_L$ roots $\alpha$.
\end{enumerate}
Then, with $R=2B-|w|^2$ and $R'=2B'-|w|^2$,
\begin{equation}
\label{eq:mainbound}
\sqrt{R}+\sqrt{R'}\ \ge\ \sqrt{2}.
\end{equation}
\end{theorem}

The proof compares the highest-weight states of the two $SU(3)_c\times SU(2)_L$ multiplets. The Standard-Model Weyl group preserves the coset in which the momenta lie and moves only their $W$-components. Because the representations that occur are minuscule, the two highest-weight states have the same $W$-projection, and the difference of their momenta is a non-zero vector of an even lattice. Nothing about the hidden sector enters.

\begin{lemma}[the Standard-Model Weyl group on the coset]
\label{lem:weyl}
Assume H1--H3. Then
\emph{(i)} $\alpha\cdot V_g\in\ZZ$ for every $SU(3)_c$ or $SU(2)_L$ root $\alpha$;
\emph{(ii)} $\sigma_\alpha$ maps the coset $V_g+\Lambda$ onto itself;
\emph{(iii)} every $\sigma\in\Wsm$ fixes $W^{\perp}$ pointwise, preserves $W$ and commutes with $P_W$, so that $|\sigma x|=|x|$ and $(\sigma x)_\perp=x_\perp$ for every $x$.
\end{lemma}

\begin{proof}
(i) By H2, $\alpha\cdot V_1$, $\alpha\cdot V_2$ and $\alpha\cdot W_a$ are integers; take the inner product of \eqref{eq:Vg} with $\alpha$, the coefficients $k,l,n_a$ being integers.

(ii) Write $p=\lambda+V_g$ with $\lambda\in\Lambda$. In $p\cdot\alpha=\lambda\cdot\alpha+V_g\cdot\alpha$ the first term is an integer because $\Lambda$ is integral by H1 and $\alpha\in\Lambda$, and the second by (i). So $\sigma_\alpha(p)=p-m\alpha$ with $m\in\ZZ$, whence $\sigma_\alpha(p)\in V_g+\Lambda$, and $\sigma_\alpha^2=\mathrm{id}$ gives surjectivity.

(iii) Each $SU(3)_c$ or $SU(2)_L$ root lies in the span of its factor's Cartan directions, hence in $W$, so for $x\perp W$ one has $x\cdot\alpha=0$ and $\sigma_\alpha(x)=x$. An orthogonal map that fixes $W^\perp$ pointwise maps $W=(W^\perp)^\perp$ to itself, so $P_W\sigma_\alpha=\sigma_\alpha P_W$, and products of such reflections inherit these properties.
\end{proof}

Fix simple roots of $SU(3)_c$ and $SU(2)_L$ and write $w=w_3+w_2+w_Y$ along $W=W_3\oplus W_2\oplus\RR t_Y$. The group $\Wsm=W_{SU(3)}\times W_{SU(2)}$ acts on $w_3$ and $w_2$ as the Weyl group of each factor, fixes $w_Y$ because $\alpha\cdot t_Y=0$, and fixes $s$ by lemma~\ref{lem:weyl}(iii). Every Weyl orbit meets the dominant chamber exactly once, so each momentum $p$ has a unique element $p^+$ of its orbit $\Wsm p$ whose $SU(3)_c$ and $SU(2)_L$ weights are dominant. For the representations at hand, $p^+$ is the momentum of the highest-weight state of the multiplet containing the state.

\begin{lemma}[highest-weight states and their difference]
\label{lem:hw}
Assume H1--H4. Then $p^+$ and $p'^+$ lie in $V_g+\Lambda$, with $|p^+|=|p|$, $|p'^+|=|p'|$, $p^+_\perp=p_\perp$, $p'^+_\perp=p'_\perp$ and $P_Wp^+=P_Wp'^+$. Consequently $\Delta:=p^+-p'^+$ equals $p_\perp-p'_\perp$, lies in $\Lambda$ and is orthogonal to $W$, and $\Delta\ne0$ if and only if H5 holds. Under H1--H5, therefore, $|\Delta|^2\ge2$.
\end{lemma}

\begin{proof}
The statements before the last equality are lemma~\ref{lem:weyl}(ii),(iii). By H4 the two states have the same $SU(3)_c$ and $SU(2)_L$ representations and the same hypercharge, so their $w_Y$ agree. The representations that occur, $\rep3$, $\arep3$, $\rep2$ and the singlet, are minuscule: their weights form a single Weyl orbit, which meets the dominant chamber in the highest weight. The $SU(3)_c$ and $SU(2)_L$ components of $P_Wp^+$ and of $P_Wp'^+$ are therefore both the highest weight of the common representation. It follows that $\Delta$ has vanishing $W$-component and $W^\perp$-component $p_\perp-p'_\perp$, and it lies in $\Lambda$ as the difference of two elements of one coset. If $p'=\sigma(p)$ with $\sigma\in\Wsm$, the two orbits coincide and $p'^+=p^+$; conversely, $\Delta=0$ means $\tau(p)=\tau'(p')$ for some $\tau,\tau'\in\Wsm$, and then $p'=\tau'^{-1}\tau(p)$. Under H5, $\Delta$ is thus a non-zero vector of the even lattice $\Lambda$, and $|\Delta|^2$ is a positive even integer.
\end{proof}

The vector $\Delta$ depends on no choice. Changing the simple roots applies one element of $\Wsm$ to both $p^+$ and $p'^+$ and leaves $\Delta\perp W$ fixed, and for any $\sigma\in\Wsm$ that brings $p'$ to the $W$-projection of $p$ the difference $p-\sigma(p')$ has vanishing $W$-component and $W^\perp$-component $p_\perp-p'_\perp$, so it equals $\Delta$.

\begin{proof}[Proof of theorem~\ref{thm:bound}]
Write $s=p_\perp$ and $s'=p'_\perp$, so that $|s|^2=R$ and $|s'|^2=R'$ by \eqref{eq:split}, and $\Delta=s-s'$ by lemma~\ref{lem:hw}. Then
\begin{equation}
\label{eq:triangle}
2\ \le\ |\Delta|^2\ =\ R+R'-2\,s\cdot s'\ \le\ R+R'+2\sqrt{RR'}\ =\ \bigl(\sqrt R+\sqrt{R'}\bigr)^2 ,
\end{equation}
by Cauchy--Schwarz applied to $-s\cdot s'$. Take square roots.
\end{proof}

A Weyl reflection is an isometry, so $p'=\sigma(p)$ forces $|p'|^2=|p|^2$, and at a common constructing element \eqref{eq:budget} then forces $\Nosc'=\Nosc$. A pair excluded by H5 is therefore one of two things: two states of one $SU(3)_c\times SU(2)_L$ multiplet, the case H5 is designed to exclude, or states of two different multiplets on the same highest-weight momentum, distinguished by different oscillator excitations of the same oscillator number or by different right-moving states. For such a pair $\Delta=0$ and the theorem is silent; section~\ref{sec:verif:h5} shows that no pair of this kind carries a Standard-Model representation in the catalogues.

For oscillator-free states \eqref{eq:budgetdiff} gives $B=B'$, hence $R=R'$, and the theorem reads
\begin{equation}
\label{eq:Rhalf}
R\ \ge\ \tfrac12 .
\end{equation}
Equality in \eqref{eq:triangle} requires $|\Delta|^2=2$ together with the Cauchy--Schwarz equality $s\cdot s'=-\sqrt{RR'}$. For equal budgets this gives $s'=-s$, $\Delta=2s$ and $4R=2$: the two states are the weights $\pm\Delta/2$ of the $SU(2)$ generated by the root $\pm\Delta$. Section~\ref{sec:ceiling} shows that this root belongs to the local gauge group and that the same structure persists away from equality.

Substituting $R=2B-|w|^2$ in \eqref{eq:Rhalf} turns the bound into a condition on the sector. Under H1--H5, two oscillator-free states of one Standard-Model representation at one constructing element require
\begin{equation}
\label{eq:threshold}
B\ \ge\ \frac{|w|^2}{2}+\frac14 ,
\qquad\text{i.e.}\qquad
B\ \ge\ \tfrac{17}{20}\ \ (\text{\ten}),
\qquad
B\ \ge\ \tfrac{13}{20}\ \ (\text{\fbar}),
\end{equation}
the two values being $\tfrac35+\tfrac14$ and $\tfrac25+\tfrac14$; the condition is necessary and not sufficient. Unlike the results of section~\ref{sec:ceiling}, it needs no restriction on $R$. Its constant is $\tfrac14=|\Delta|^2_{\min}/8$ with $|\Delta|^2_{\min}=2$, which is also the smallest non-zero conformal weight of a level-one primary of a simply-laced affine algebra, attained by the $A_1$ doublet. At level one the primaries are the fundamental weights of unit comark, with $h=|\lambda|^2/2$; the smallest value is $n/2(n+1)\ge\tfrac14$ for $A_n$, with equality at $n=1$, $\tfrac12$ for the vector of $D_n$ and $n/8\ge\tfrac12$ for its spinor at $n\ge4$, $\tfrac23$ for $E_6$ and $\tfrac34$ for $E_7$, and $E_8$ has no non-trivial level-one primary. The gap between the two thresholds, $\tfrac{17}{20}-\tfrac{13}{20}=\tfrac15$, is $h(\rep{10})-h(\rep5)$.

%%%%%%%%%%%%%%%%%%%%%%%%%%%%%%%%%%%%%%%%%%%%%%%%%%%%%%%%%%%%%%%%%%%%%%%%%%%%%%%
\section{The local \texorpdfstring{$SU(n)$}{SU(n)} and the multiplicity ceiling}
\label{sec:ceiling}

Applied to all pairs of a family at once, the argument of section~\ref{sec:bound} determines how the family is organised. Write
\begin{equation}
\label{eq:HVg}
H(V_g)=\{\alpha\in\Lambda:\ |\alpha|^2=2,\ \alpha\cdot V_g\in\ZZ\}
\end{equation}
for the roots of the local gauge group at $g$, the group left unbroken by the local shift \eqref{eq:Vg} alone, Wilson lines included.

\begin{proposition}[the local $SU(n)$]
\label{prop:sun}
Let $p_1,\dots,p_n$, $n\ge2$, be the gauge momenta of massless left-chiral states of one Standard-Model representation, all built on one constructing element $g$, with budgets $B_i$, oscillator numbers $N_{{\rm osc},i}$ and $R_i=2B_i-|w|^2$. Assume H1--H4, H5 for every pair, and
\begin{equation}
\label{eq:sundomain}
\sqrt{R_i}+\sqrt{R_j}<2\qquad(i\ne j),
\end{equation}
which holds in particular when every $R_i<1$. Let $q_i=p_i^+$, $\Delta_{ij}=q_i-q_j$ and $c=\tfrac1n\sum_k q_k$. Then:
\begin{enumerate}
\item[(a)] every $\Delta_{ij}$ is a root of $\Lambda$ orthogonal to $W$;
\item[(b)] $q_i\cdot\Delta_{ij}=1+B_i-B_j=1+N_{{\rm osc},j}-N_{{\rm osc},i}$, and therefore
  $\Delta_{ij}\cdot V_g\equiv N_{{\rm osc},j}-N_{{\rm osc},i} \pmod 1$;
\item[(c)] $q_i\cdot q_j=B_i+B_j-1$ for $i\ne j$; the vectors $\alpha_m=q_m-q_{m+1}$,
  $m=1,\dots,n-1$, have the Cartan matrix of $A_{n-1}$, and the $\Delta_{ij}$ are the
  $n(n-1)$ roots of the $A_{n-1}$ root system they span; the vectors $\mu_i=q_i-c$ satisfy
  $\mu_i\cdot\alpha_m=\delta_{im}-\delta_{i,m+1}$, so they are the $n$ weights of the
  fundamental representation, $\mu_1$ its highest weight, and $|\mu_i|^2=(n-1)/n$.
\end{enumerate}
If in addition $N_{{\rm osc},i}-N_{{\rm osc},j}\in\ZZ$ for every pair, in particular if every state is
oscillator-free, then:
\begin{enumerate}
\item[(d)] every $\Delta_{ij}$ is a root of the local group $H(V_g)$; the $\Delta_{ij}$ are
  the roots of a subgroup $SU(n)$ of the local gauge group which commutes with
  $SU(3)_c\times SU(2)_L\times U(1)_Y$, and the $n$ states carry the $n$ weights of its
  fundamental representation;
\item[(e)] if the states are oscillator-free, $c=w+c_\perp$ with $c_\perp$ orthogonal to $W$
  and to every $\Delta_{ij}$ and common to all $n$ states, and
  \begin{equation}
  \label{eq:sunsplit}
  R=\frac{n-1}{n}+|c_\perp|^2 ,\qquad\text{i.e.}\qquad
  \frac R2=h_n+\frac{|c_\perp|^2}{2},\quad h_n=\frac{n-1}{2n},
  \end{equation}
  where $h_n$ is the conformal weight of the fundamental representation of $SU(n)$ at level
  one.
\end{enumerate}
\end{proposition}

\begin{proof}
(a) By lemma~\ref{lem:hw}, applied to each pair, the $q_i$ share one $W$-projection $w$ and lie in $V_g+\Lambda$, and $\Delta_{ij}=s_i-s_j$ is a non-zero element of $\Lambda$ orthogonal to $W$, where $s_i=(p_i)_\perp=(q_i)_\perp$ and $|s_i|^2=R_i$. Now $|\Delta_{ij}|^2=R_i+R_j-2\,s_i\cdot s_j$, and Cauchy--Schwarz bounds $-s_i\cdot s_j$ above by $\sqrt{R_iR_j}$, so $|\Delta_{ij}|^2\le(\sqrt{R_i}+\sqrt{R_j})^2<4$ by \eqref{eq:sundomain}. By H1 the squared length is a positive even integer, hence $2$.

(b) $|q_i|^2=2B_i$ by lemma~\ref{lem:hw} and \eqref{eq:budget}. Expanding $|q_j|^2=|q_i-\Delta_{ij}|^2$ with $|\Delta_{ij}|^2=2$ gives $q_i\cdot\Delta_{ij}=1+B_i-B_j$, and $B_i-B_j=N_{{\rm osc},j}-N_{{\rm osc},i}$ by \eqref{eq:budgetdiff}. Write $q_i=V_g+\lambda_i$ with $\lambda_i\in\Lambda$. Then $\Delta_{ij}\cdot V_g=q_i\cdot\Delta_{ij}-\lambda_i\cdot\Delta_{ij}$, and the second term is an integer because $\Lambda$ is integral.

(c) $|q_i-q_j|^2=2$ and $|q_i|^2=2B_i$ give $q_i\cdot q_j=B_i+B_j-1$. For any four indices, $(q_a-q_b)\cdot(q_c-q_d)=\tfrac12\bigl(|q_a-q_d|^2+|q_b-q_c|^2-|q_a-q_c|^2-|q_b-q_d|^2\bigr)$, and every distance between distinct indices is $\sqrt2$. Hence $\alpha_m\cdot\alpha_m=2$, $\alpha_m\cdot\alpha_{m+1}=-1$ and $\alpha_m\cdot\alpha_k=0$ for $|m-k|\ge2$, which is the Cartan matrix of $A_{n-1}$. It is positive definite, so the $\alpha_m$ are independent, and $\Delta_{ij}=\alpha_i+\dots+\alpha_{j-1}$ for $i<j$: together with their negatives these are the roots of $A_{n-1}$ in the base $\{\alpha_m\}$. From $q_i\cdot q_j=B_i+B_j-1$ one finds $q_i\cdot\alpha_m=B_m-B_{m+1}+\delta_{im}-\delta_{i,m+1}$, and summing over $i$, $c\cdot\alpha_m=B_m-B_{m+1}$; hence $\mu_i\cdot\alpha_m=\delta_{im}-\delta_{i,m+1}$. Since $\mu_i=\tfrac1n\sum_k\Delta_{ik}$ lies in the span of the roots, these pairings determine $\mu_i$. They are the pairings of the weights $\varepsilon_i-\tfrac1n\sum_k\varepsilon_k$ of the fundamental representation with the simple roots $\varepsilon_m-\varepsilon_{m+1}$, the Dynkin labels of $\mu_1$ are $(1,0,\dots,0)$, and $|\mu_i|^2=\tfrac1{n^2}\sum_{k,l}\Delta_{ik}\cdot\Delta_{il}=\bigl(2(n-1)+(n-1)(n-2)\bigr)/n^2=(n-1)/n$, using $\Delta_{ik}\cdot\Delta_{il}=1$ for distinct $k,l\ne i$.

(d) By (b), $\Delta_{ij}\cdot V_g\in\ZZ$, so every $\Delta_{ij}$ lies in $H(V_g)$. The set $\{\Delta_{ij}\}$ is closed in the root system of $\Lambda$: the sum of two of its elements is either again one of them, $\Delta_{ij}+\Delta_{jk}=\Delta_{ik}$, or is zero, or has squared length $4$ or $6$ and is no root. It is therefore the root system of a regular subalgebra $su(n)$ of the local algebra. For a root $\beta$ of $SU(3)_c$ or $SU(2)_L$ one has $\beta\in W$ and $\Delta_{ij}\perp W$, so $\beta\cdot\Delta_{ij}=0$ and $|\beta\pm\Delta_{ij}|^2=4$, and the root vectors commute; every element of $W$, $t_Y$ included, annihilates $\Delta_{ij}$. The states carry the weights of (c).

(e) For equal budgets (c) gives $c\cdot\alpha_m=0$, so $c$ is orthogonal to the span $U$ of the roots. Every $q_i$ has $W$-projection $w$ and $\mu_i\in U\perp W$, so $P_Wc=w$ and $c=w+c_\perp$ with $c_\perp\perp W\oplus U$. Then $s_i=c_\perp+\mu_i$ is an orthogonal sum, and $R=|s_i|^2=|c_\perp|^2+(n-1)/n$.
\end{proof}

When two states differ in oscillator number by a fraction, their difference is still a root of $\Lambda$ orthogonal to $W$, but by (b) it pairs with the local shift to that fraction modulo one and is no root of the local group. When $\sqrt{R_i}+\sqrt{R_j}\ge2$ the chain gives only $|\Delta_{ij}|^2\le(\sqrt{R_i}+\sqrt{R_j})^2$, which admits $|\Delta_{ij}|^2=4$, and no root need appear.

\begin{corollary}[multiplicity ceiling]
\label{cor:ceiling}
Let $p_1,\dots,p_n$, $n\ge2$, be the momenta of oscillator-free massless left-chiral states of one Standard-Model representation, all built on one constructing element, with H1--H4 in force, H5 holding for every pair, and common value $R=2B-|w|^2<1$. Then
\begin{equation}
\label{eq:ceiling}
n\ \le\ \frac{1}{1-R},
\end{equation}
with equality if and only if $c_\perp=0$.
\end{corollary}

\begin{proof}
The budgets are equal by \eqref{eq:budgetdiff}, so \eqref{eq:sundomain} reduces to $R<1$, and \eqref{eq:sunsplit} gives $1-\tfrac1n=R-|c_\perp|^2\le R$, with equality exactly when $c_\perp=0$; that is, $n\le1/(1-R)$.
\end{proof}

In terms of the Gram matrix of the $s_i$, which part (a) makes $\mathbf{1}+(R-1)J$ with $J$ the all-ones matrix, the ceiling is positivity: the eigenvalue $1+n(R-1)$ on $(1,\dots,1)$ equals $n|c_\perp|^2$, since $\sum_is_i=nc_\perp$. The $n$ states spend $h_n$ of their budget on the fundamental weight of the local $SU(n)$ and the rest, $|c_\perp|^2/2$, on a charge they share.

For two states the weights are $\mu=\pm\Delta/2$ with $|\mu|^2=\tfrac12$, the conformal weight $h_2=\tfrac14$ is the constant of \eqref{eq:threshold}, and $|c_\perp|^2=R-\tfrac12$. The equality case of \eqref{eq:Rhalf} is $c_\perp=0$: two states saturating the bound are a doublet of an $SU(2)$ of the local group and carry no other charge. Away from saturation the difference of two oscillator-free states of one representation at one element, with $R<1$, is still a root of $H(V_g)$; saturation adds only the absence of the common charge.

The ceiling has content only for $R<1$: for $R\ge1$ the bound \eqref{eq:Rhalf} holds trivially and \eqref{eq:ceiling} is vacuous. In the two catalogues of section~\ref{sec:verif} every realised budget has $R<1$, the largest value being $R=83/90$, on a \fbar{} representation. Below threshold the ceiling reproduces the threshold, since $1/(1-R)<2$ exactly when $R<\tfrac12$, so a representation below its threshold cannot repeat at all. The ceiling does not need a common budget. For $n$ states of one representation at one constructing element, with H1--H5 for every pair, lemma~\ref{lem:hw} gives $|s_i-s_j|^2\ge2$ for $i\ne j$, where $s_i=(p_i)_\perp$, and the identity $\sum_{i<j}|s_i-s_j|^2=n\sum_iR_i-|\sum_is_i|^2$ then yields
\begin{equation}
\label{eq:meanceiling}
\bar R\ \ge\ \frac{n-1}{n}+|\bar s|^2,\qquad \bar R=\frac1n\sum_iR_i,\quad \bar s=\frac1n\sum_is_i ,
\end{equation}
so that $n\le1/(1-\bar R)\le1/(1-\max_iR_i)$ whenever $\bar R<1$, with or without oscillators. Under \eqref{eq:sundomain} every $|s_i-s_j|^2$ equals $2$ and \eqref{eq:meanceiling} is an equality, $\bar s$ being $c_\perp$; for equal budgets it is \eqref{eq:sunsplit}. For two states \eqref{eq:mainbound} is the stronger condition on the budgets, since $(\sqrt R+\sqrt{R'})^2\le2(R+R')$.

The ceiling is not the only restriction. The roots $\alpha_m$ of proposition~\ref{prop:sun}(c) form a connected $A_{n-1}$ diagram, and a root of $\Lambda=\Lambda_{E_8}\oplus\Lambda_{E_8}$ lies in one summand, so all of them lie in one $E_8$ factor, adjacent ones not being orthogonal. The largest subsystem of type $A$ in $E_8$ is $A_8$; hence $n\le9$ under the hypotheses of the proposition, whatever $R$. If the Standard Model sits in an $SU(5)$ of one factor whose Cartan subalgebra is $W$, the roots of that factor orthogonal to $W$ are those of the commuting $SU(5)$, and a family whose roots lie there has $n\le5$.

A root of the local group need not be a root of the unbroken four-dimensional group, which requires in addition an integral pairing with every other combination of shifts and Wilson lines. The four-dimensional condition is strictly stronger at the great majority of the twisted constructing elements we surveyed; section~\ref{sec:signature} gives the count, together with the exception where the two coincide.

%%%%%%%%%%%%%%%%%%%%%%%%%%%%%%%%%%%%%%%%%%%%%%%%%%%%%%%%%%%%%%%%%%%%%%%%%%%%%%%
\section{Multiplicity saturation in two orbifold catalogues}
\label{sec:verif}

How closely do actual spectra approach the ceiling? We use the complete massless spectra of $16{,}138$ models drawn from two published catalogues of Standard-Model-like models \cite{Parr:2020oar}, one on $T^6/(\ZZ_2\times\ZZ_4)$ and one on $T^6/(\ZZ_3\times\ZZ_6)$, on the space groups of the \textsc{orbifolder} geometry files \texttt{Geometry\_Z2xZ4\_1\_1} (lattice $\ZZ^6$, Wilson lines $W_1=W_2$, $W_3$, $W_4$, $W_5=W_6$ of order two) and \texttt{Geometry\_Z3xZ6\_1\_1} (section~\ref{sec:construction:model}), two of the symmetric toroidal orbifolds of the classification of ref.~\cite{Fischer:2012qj}; we recompute every spectrum with the \textsc{orbifolder}. We count models; multiplets; \emph{populated fixed points}, meaning constructing elements of a model at which massless states sit; Standard-Model representations at a populated fixed point, whose multiplicity is the number of multiplets carrying them there; and same-representation pairs. A populated fixed point can carry several representations, so representations at fixed points outnumber populated fixed points. The counting conventions are the same on both geometries and are collected in appendix~\ref{app:defs}.

Table~\ref{tab:verif} collects the catalogue-wide counts and table~\ref{tab:hist} the multiplicities behind them. A multiplet of \emph{hidden dimension} $d$, the dimension of its representation under the gauge factors other than those of the Standard Model, contains $d$ states of equal $|p|^2$, and where the distinction matters the last column of table~\ref{tab:verif} counts states rather than multiplets.\footnote{In the row of populated fixed points at which \eqref{eq:threshold} allows a \ten{} repeat it counts the \ten{} states that meet the threshold, and in the pair row unordered pairs of distinct states, $35{,}491$ of which lie across distinct multiplets.} The hypercharge generator of each model is orthogonal to the $SU(3)_c$ and $SU(2)_L$ roots, and the momentum of each multiplet has the length that \eqref{eq:budget} assigns to its sector and oscillator number.

\begin{table}[t]
\centering
\footnotesize
\setlength{\tabcolsep}{2.5pt}
\begin{tabular}{lrrrr}
\toprule
quantity & $\ZZ_2\times\ZZ_4$ & $\ZZ_3\times\ZZ_6$ & both & states\\
\midrule
models & $11{,}526$ & $4{,}612$ & $16{,}138$ & ---\\
populated twisted fixed points & $584{,}855$ & $295{,}734$ & $880{,}589$ & ---\\
twisted constructing elements & $58$ & $70$ & $128$ & ---\\
multiplets in a Standard-Model representation & $328{,}784$ & $173{,}612$ & $502{,}396$ & $518{,}373$\\
\quad \ten & $72{,}494$ & $37{,}223$ & $109{,}717$ & $109{,}717$\\
\quad \fbar & $256{,}290$ & $136{,}389$ & $392{,}679$ & $408{,}656$\\
\midrule
\ten{} representations at a fixed point & $72{,}494$ & $37{,}223$ & $109{,}717$ & ---\\
\quad multiplicity & all at $1$ & all at $1$ & all at $1$ & all at $1$\\
\fbar{} representations at a fixed point & $231{,}323$ & $128{,}747$ & $360{,}070$ & ---\\
\quad all members below threshold, $R<1/2$ & $3{,}913$ & $653$ & $4{,}566$ & all at $1$\\
populated fixed points with a \ten{} representation & $29{,}242$ & $16{,}208$ & $45{,}450$ & ---\\
\midrule
populated fixed points where \eqref{eq:threshold} allows a \ten{} repeat & $0$ & $65$ & $65$ & $65$\\
same-representation pairs of distinct multiplets & $25{,}883$ & $7{,}860$ & $33{,}743$ & $51{,}879$\\
\midrule
groups of $\ge2$ multiplets on one gauge momentum & $18{,}844$ & $4{,}775$ & $23{,}619$ & ---\\
\quad populated fixed points hosting one & $14{,}092$ & $4{,}406$ & $18{,}498$ & ---\\
\quad carrying a Standard-Model representation & $0$ & $0$ & $0$ & $0$\\
\bottomrule
\end{tabular}
\caption{Catalogue-wide counts. Counting conventions and the definitions of the rows are given in appendix~\ref{app:defs}; \ten{} fixed points are counted in either chirality.}
\label{tab:verif}
\end{table}

\begin{table}[t]
\centering
\footnotesize
\begin{tabular}{crrrrrr}
\toprule
 & \multicolumn{2}{c}{\ten, multiplets} & \multicolumn{2}{c}{\fbar, multiplets} & \multicolumn{2}{c}{\fbar, states}\\
\cmidrule(lr){2-3}\cmidrule(lr){4-5}\cmidrule(lr){6-7}
multiplicity & $\ZZ_2\times\ZZ_4$ & $\ZZ_3\times\ZZ_6$ & $\ZZ_2\times\ZZ_4$ & $\ZZ_3\times\ZZ_6$ & $\ZZ_2\times\ZZ_4$ & $\ZZ_3\times\ZZ_6$\\
\midrule
$1$ & $72{,}494$ & $37{,}223$ & $207{,}270$ & $121{,}319$ & $197{,}193$ & $117{,}494$\\
$2$ & --- & --- & $23{,}141$ & $7{,}218$ & $31{,}811$ & $10{,}458$\\
$3$ & --- & --- & $910$ & $206$ & $2{,}238$ & $788$\\
$4$ & --- & --- & $2$ & $4$ & $81$ & $6$\\
$5$ & --- & --- & --- & --- & --- & $1$\\
\bottomrule
\end{tabular}
\caption{Multiplicities of the Standard-Model representations at populated fixed points.}
\label{tab:hist}
\end{table}

Since \eqref{eq:mainbound} is a theorem, every pair obeys it, counted in multiplets or in states; what the spectra add is that no \ten{} representation repeats anywhere in either catalogue, again in multiplets or in states. On $\ZZ_2\times\ZZ_4$ this is forced, since every twisted sector has $B\le13/16<17/20$; on $\ZZ_3\times\ZZ_6$ the threshold \eqref{eq:threshold} is met at $65$ populated fixed points, to which we return in section~\ref{sec:construction:classes}.

The same-representation pairs of the catalogues lie inside the domain \eqref{eq:sundomain} of proposition~\ref{prop:sun}: over the $33{,}743$ pairs of distinct multiplets the largest value of $\sqrt R+\sqrt{R'}$ is $1.8166$, so each pair differs by a root of $\Lambda$ orthogonal to $W$. That root lies in $H(V_g)$ for the $33{,}537$ pairs of equal oscillator number, while for the $206$ pairs whose oscillator numbers differ by $\tfrac14$ (on $\ZZ_2\times\ZZ_4$) or $\tfrac16$ (on $\ZZ_3\times\ZZ_6$) it pairs with the local shift to that fraction modulo one. The $1{,}122$ groups of multiplicity three and four have the Cartan matrix of $A_{n-1}$, with every root in $H(V_g)$.

\begin{table}[t]
\centering
\small
\begin{tabular}{lccccc}
\toprule
geometry & $R$ & $B$ & $|w|^2$ & ceiling $\lfloor1/(1-R)\rfloor$ & largest multiplicity found\\
\midrule
$\ZZ_3\times\ZZ_6$ & $83/90$ & $31/36$ & $4/5$ & $12$ ($9$) & $2$\\
$\ZZ_2\times\ZZ_4$ & $33/40$ & $13/16$ & $4/5$ & $5$ & $4$\\
$\ZZ_3\times\ZZ_6$ & $34/45$ & $7/9$   & $4/5$ & $4$ & $4$\\
both               & $7/10$  & $3/4$   & $4/5$ & $3$ & $3$\\
$\ZZ_3\times\ZZ_6$ & $53/90$ & $25/36$ & $4/5$ & $2$ & $2$\\
$\ZZ_2\times\ZZ_4$ & $23/40$ & $11/16$ & $4/5$ & $2$ & $2$\\
$\ZZ_3\times\ZZ_6$ & $8/15$  & $2/3$   & $4/5$ & $2$ & $1$\\
$\ZZ_3\times\ZZ_6$ & $47/90$ & $31/36$ & $6/5$ & $2$ & $1$\\
$\ZZ_2\times\ZZ_4$ & $17/40$ & $13/16$ & $6/5$ & $1$ & $1$\\
$\ZZ_3\times\ZZ_6$ & $16/45$ & $7/9$   & $6/5$ & $1$ & $1$\\
both               & $3/10$  & $3/4$   & $6/5$ & $1$ & $1$\\
$\ZZ_2\times\ZZ_4$ & $7/40$  & $11/16$ & $6/5$ & $1$ & $1$\\
$\ZZ_3\times\ZZ_6$ & $17/90$ & $25/36$ & $6/5$ & $1$ & $1$\\
$\ZZ_2\times\ZZ_4$ & $13/40$ & $9/16$  & $4/5$ & $1$ & $1$\\
$\ZZ_2\times\ZZ_4$ & $3/40$  & $7/16$  & $4/5$ & $1$ & $1$\\
$\ZZ_3\times\ZZ_6$ & $11/30$ & $7/12$  & $4/5$ & $1$ & $1$\\
$\ZZ_3\times\ZZ_6$ & $23/90$ & $19/36$ & $4/5$ & $1$ & $1$\\
$\ZZ_3\times\ZZ_6$ & $4/45$  & $4/9$   & $4/5$ & $1$ & $1$\\
\bottomrule
\end{tabular}
\caption{The ceiling \eqref{eq:ceiling} and the largest same-representation multiplicity found at each realised value of $R$; in brackets, with the rank restriction of section~\ref{sec:ceiling}.}
\label{tab:ceiling}
\end{table}

Table~\ref{tab:ceiling} sets the ceiling against the largest multiplicity found at each realised value of $R$. A group of multiplets of one representation at one fixed point is listed under the largest $R$ among its members. Groups of mixed budget occur in two rows only. They obey \eqref{eq:meanceiling}, at mean budgets $\bar R=23/40$ and $34/45$, which allow two and four states, and \eqref{eq:mainbound}, with $(\sqrt{33/40}+\sqrt{13/40})^2=2.186$ and $(\sqrt{83/90}+\sqrt{53/90})^2=2.985$. On $\ZZ_2\times\ZZ_4$, $197$ \fbar{} multiplets at $R=13/40$ sit in groups whose largest $R$ is $33/40$; they do not change that row's largest multiplicity, which is four. On $\ZZ_3\times\ZZ_6$, the nine pairs of section~\ref{sec:signature} each have one member at $R=83/90$ and one at $53/90$, and they are the only groups of two in the $R=83/90$ row; no single-budget group there repeats. Every other row has a single common budget, and counted in states rather than multiplets every row again lies within its ceiling.

The ceiling is reached. At four values of $R$, on both geometries, the largest multiplicity equals $\lfloor1/(1-R)\rfloor$: two at $R=23/40$ and $53/90$, three at $R=7/10$ and four at $R=34/45$, in multiplets and in states alike. At none of these values is $1/(1-R)$ an integer, so by proposition~\ref{prop:sun}(e) every saturating family carries a non-zero common charge $c_\perp$. Elsewhere the ceiling leaves room: at $R=33/40$ it allows five and four occur, and at $R=83/90$ the norm allows twelve and the rank restriction nine, while no single-budget group repeats.

Two states that are both below threshold cannot form a pair, with or without oscillators, since then $\sqrt R+\sqrt{R'}<\sqrt2$. A state below threshold can still pair with a state of higher budget at the same element, as the $197$ \fbar{} multiplets at $R=13/40$ above do; what excludes a partner is the largest budget available there. On $\ZZ_2\times\ZZ_4$ the sector $T(1,3)$ has oscillator-free budget $11/16$, and the smallest excitation, $\Nosc=1/4$, lowers it to $B=7/16$, so a \fbar{} representation there has $R=2B-\tfrac45=\tfrac3{40}$. Corollary~\ref{cor:ceiling} does not apply to such states, but theorem~\ref{thm:bound} does: every state at this element has $B\le\tfrac{11}{16}$ and hence $R\le\tfrac{23}{40}$, and $\sqrt{3/40}+\sqrt{23/40}=1.032<\sqrt2$, so a state at $R=3/40$ has no partner of its representation there and $n\le1$. Every group of \fbar{} representations whose members are all below threshold has a single member (table~\ref{tab:verif}), while the same representations reach multiplicity four at higher budget: the multiplicity follows $2B-|w|^2$ and not the representation.

\subsection{Coincident momenta and hidden multiplets}
\label{sec:verif:h5}

Two or more multiplets on one gauge momentum at one fixed point, distinguished only by their oscillator excitations, do occur (table~\ref{tab:verif}), and by section~\ref{sec:bound} their oscillator numbers are equal. None of them carries a Standard-Model representation in either catalogue, so the case excluded by H5 does not arise there. For the \ten{} representations on $\ZZ_2\times\ZZ_4$ this is forced: they need $B\ge3/5$, while the largest sector budget less the smallest oscillator mode is $13/16-1/4=9/16$, so no \ten{} state there carries an oscillator. On $\ZZ_3\times\ZZ_6$, where eleven oscillator-carrying \ten{} multiplets occur, it is a property of the spectra alone.

A multiplet of hidden dimension $d$ contains $d$ gauge momenta, and a pairwise test over multiplets does not compare them with one another. For the \ten{} representations this leaves nothing out: every \ten{} multiplet in either catalogue has hidden dimension one, so no \ten{} representation is doubled inside a multiplet, and the same holds for every multiplet below either threshold. Hidden dimension counts multiplicity of the Standard-Model representation only if the hidden factors commute with the Standard Model, so that their roots are orthogonal to $W$. This follows from the identification of the Standard Model, since $t_Y$ is orthogonal to every simple root of the four-dimensional group and distinct simple factors are orthogonal; in all $16{,}138$ models every root of the four-dimensional gauge group lies in $W$ or is orthogonal to it.

%%%%%%%%%%%%%%%%%%%%%%%%%%%%%%%%%%%%%%%%%%%%%%%%%%%%%%%%%%%%%%%%%%%%%%%%%%%%%%%
\section{An explicit model at the bound}
\label{sec:construction}

\subsection{The model}
\label{sec:construction:model}

On $T^6/(\ZZ_3\times\ZZ_6)$ with twists $v_1=(0,\tfrac13,-\tfrac13)$ and $v_2=(\tfrac16,0,-\tfrac16)$, the sectors $T(0,1)$ and $T(0,5)$ have $\delta c=5/36$ and hence the oscillator-free budget $B=31/36$, above the \ten{} threshold $17/20$. We use the space group of the \textsc{orbifolder} geometry file \texttt{Geometry\_Z3xZ6\_1\_1}, whose lattice is the root lattice of $SU(3)^3$, the torus $T^2_a$ being spanned by $e_{2a-1}=(2,0)$ and $e_{2a}=(-1,\sqrt3)$. It allows a non-trivial Wilson line only along the torus that $\omega$ leaves fixed, of order three, and identifies the two lines along the basis vectors of that torus, which $\theta$ maps into each other, so that $W_4=W_3$; the Wilson lines $W_1,W_2,W_5,W_6$ along the other two tori are trivial, and we set them to zero. We built a modular-invariant model there, in the standard basis and without discrete torsion, from an order-three shift $V_1$, an order-six shift $V_2$ and the Wilson line $W_3=W_4$:
\begin{align}
\label{eq:model}
V_1&=\bigl(\tfrac12,\tfrac32,\tfrac12,\tfrac16,\tfrac76,\tfrac12,\tfrac12,\tfrac12\,\bigl|\,0^8\bigr),\nonumber\\
V_2&=\bigl(\tfrac23,\tfrac23,-\tfrac13,-\tfrac13,-\tfrac13,0,0,0\,\bigl|\,\tfrac12,\tfrac12,0,0,0,0,0,0\bigr),\\
W_3=W_4&=\bigl(\tfrac23,-\tfrac43,\tfrac23,0,0,0,0,0\,\bigl|\,\tfrac13,-\tfrac13,\tfrac23,0,0,0,0,0\bigr),\nonumber\\
t_Y&=\bigl(-\tfrac13,-\tfrac13,-\tfrac13,\tfrac12,\tfrac12,0,0,0\,\bigl|\,0^8\bigr).\nonumber
\end{align}
These are the vectors used, not only their classes modulo $\Lambda$: shifting $V_1$ or $W_3$ by a lattice vector can change the spectrum (section~\ref{sec:construction:classes}). Its four-dimensional gauge group contains $SU(3)_c\times SU(2)_L\times U(1)_Y$ as whole factors, together with $SU(4)$, $E_6\times SU(2)$ and two further $U(1)$ factors, and the hypercharge is at $k_Y=5/3$. Every abelian anomaly trace vanishes, so the model has no anomalous $U(1)$; in particular $\Tr Y=\Tr Y^3=0$, the traces running over the left-chiral massless multiplets weighted by dimension. The net number of quark doublets, with the colour representation of each state read from its gauge momentum, is $+3$. With this $t_Y$, $W$ is the traceless part of the first five coordinates, and the weights of the \ten{} family, written on the first five coordinates with zeros in the remaining eleven as in appendix~\ref{app:classes:param}, are
\begin{equation}
\label{eq:wij}
w_{ij}=(e_i+e_j)-\tfrac25(1,1,1,1,1),\qquad 1\le i<j\le5,
\end{equation}
together with their negatives for the conjugate representations.

\begin{table}[t]
\centering
\small
\begin{tabular}{ll}
\toprule
sector and fixed point & $T(0,1)$, $(0^6)$\\
gauge momenta & $p_{1,2}=\bigl(-\tfrac13,-\tfrac13,\tfrac23,-\tfrac13,\tfrac23,0,0,0\,\bigl|\,\mp\tfrac12,\mp\tfrac12,0^6\bigr)$\\
$\Delta$, $|\Delta|^2$ & $\bigl(0^8\,\bigl|\,1,1,0^6\bigr)$, $2$\\
$p\cdot t_Y$ & $+\tfrac16$, for both\\
$SU(3)_c$, $SU(2)_L$ weights & colour $e_3$, weak $e_5$: the \ten{} weight $w_{35}$, whence $B_v=-2$\\
representation, as listed & $(\rep1,\arep3,\rep2,\rep1,\rep2)$\\
oscillator & none\\
hidden dimension & $2$\\
Standard-Model representation & $Q=(\rep3,\rep2)_{+1/6}$\\
$|p|^2$, $B$ & $31/18$, $31/36$\\
$|w|^2$, $R$ & $6/5$, $47/90$\\
$|p\cdot\Delta|$ & $1$\\
\bottomrule
\end{tabular}
\caption{The doubled fixed point of the model \eqref{eq:model}. The weights $w_{ij}$ are defined in \eqref{eq:wij} and $B_v$ in appendix~\ref{app:classes:param}.}
\label{tab:doubled}
\end{table}

Table~\ref{tab:doubled} shows the doubled fixed point. The two momenta have the same $W$-projection, so no Weyl reflection is needed to compare them and $\Delta$ is simply their difference, and the colour weight $e_3$ together with $p\cdot t_Y=+\tfrac16$ identifies the state as the $Q$ of the \ten{} family.\footnote{The \textsc{orbifolder} orders the simple roots of $SU(3)_c$ as $(e_2-e_3,\,e_1-e_2)$, the reverse of the ordering in \eqref{eq:model}, and therefore lists this triplet as $\arep3$. Reversing the order of the Dynkin labels of $SU(3)$ is conjugation, and the same convention holds for all $53$ colour-triplet multiplets of the model.}

The doubling sits just above the bound. $B=31/36$ exceeds the threshold $17/20$ by $1/90$, so $R$ exceeds its threshold by $1/45$, and the bound is cleared narrowly, $2\sqrt{47/90}=1.4453$ against $\sqrt2=1.4142$. The ceiling $1/(1-47/90)=90/43=2.093$ gives $n\le2$, and exactly two states are realised. By proposition~\ref{prop:sun} they form a doublet of an $SU(2)$ in the local gauge group, and the margin $1/45$ in $R$ is the squared length $|c_\perp|^2$ of the charge they share.

At this fixed point the hypotheses of theorem~\ref{thm:bound} take a simple form. The sector has a single constructing element, which carries no translation (section~\ref{sec:signature}), so $V_g=V_2$ exactly, with no Wilson-line contribution and no dependence on the fixed point. The Standard-Model factors are factors of the four-dimensional gauge group, and every $SU(3)_c$ and $SU(2)_L$ root lies in the first $E_8$ factor, while the two momenta, which share their representation and hence $|w|^2$, differ by a root of $\Lambda$, $|\Delta|^2=2$, lying in the second (table~\ref{tab:doubled}); $\Wsm$ therefore acts trivially on the block in which they differ.

We computed the spectrum of the model with the \textsc{orbifolder} \cite{Nilles:2011aj}, and again with the \textsc{orbifolder} from an independently written input deck; the two runs agree on all $147$ multiplets, on the gauge group, its abelian generators and simple roots, and on the anomaly traces. An independently written code for the orbifold projection, applied to the data \eqref{eq:model} at all $140$ pairs of a twisted constructing element and a massless right-mover of either chirality, reproduces the $88$ oscillator-free twisted multiplets of the spectrum at the $44$ pairs that carry them, with their CPT conjugates in the opposite chirality, and finds no state at any other pair; it does not compute states with oscillators. From the same data it finds the same four-dimensional gauge group. The model shows that the doubling exists; we do not regard it as realistic or as a candidate vacuum.

\subsection{The two shift classes}
\label{sec:construction:classes}

On this geometry and at this budget, the order-six shifts that allow a doubled \ten{} representation can be listed completely. We say that an order-six shift class $V_2+\Lambda$ \emph{admits a doubling} if there exist vectors $p\in V_2+\Lambda$ and $\Delta\in\Lambda$ with
\begin{equation}
\label{eq:admits}
|p|^2=\tfrac{31}{18},\qquad P_Wp\in\{\pm w_{ij}\},\qquad |\Delta|^2=2,\qquad \Delta\perp W,\qquad p\cdot\Delta=\pm1 ;
\end{equation}
the partner $p'=p-\mathrm{sgn}(p\cdot\Delta)\,\Delta$ then has $|p'|^2=|p|^2$ and $P_Wp'=P_Wp$. The condition involves only the lattice. It is necessary for a doubled \ten{} representation at one constructing element; whether the two states survive the orbifold projection depends on the rest of the model.

Two order-six shift vectors are \emph{equivalent} when they differ by a lattice translation followed by an automorphism of $\Lambda$ that preserves each $E_8$ factor, carries the Standard-Model root system $\{\pm(e_i-e_j)\}_{1\le i<j\le3}\cup\{\pm(e_4-e_5)\}$ to itself and carries $t_Y$ to $\pm t_Y$. The condition is setwise rather than pointwise, because the reflections in the Standard-Model simple roots permute the weights inside a representation, and the global $-1$, which is the relabelling $\omega\mapsto\omega^{-1}$ together with charge conjugation, sends $t_Y$ to $-t_Y$.

\begin{theorem}[the admitting classes]
\label{thm:twoclasses}
Assume:
\begin{enumerate}
\item[\textbf{S1}] the geometry is $T^6/(\ZZ_3\times\ZZ_6)$ in the realisation whose twists are $v_1=(0,\tfrac13,-\tfrac13)$ and $v_2=(\tfrac16,0,-\tfrac16)$, and the doubling sits in $T(0,1)$ or $T(0,5)$, where $B=31/36$;
\item[\textbf{S2}] both states have $\Nosc=0$, hence $|p|^2=31/18$;
\item[\textbf{S3}] the Standard Model sits inside the $SU(5)$ of the first $E_8$ factor, with $SU(3)_c=\langle e_1-e_2,e_2-e_3\rangle$, $SU(2)_L=\langle e_4-e_5\rangle$ and $t_Y$ as in \eqref{eq:model}, and $SU(3)_c\times SU(2)_L\times U(1)_Y$ survives in four dimensions;
\item[\textbf{S4}] two distinct massless left-chiral states of one \ten{} representation are built on one and the same constructing element and are distinct modulo $\Wsm$.
\end{enumerate}
Then exactly two inequivalent order-six shift classes admit a doubling in the sense of \eqref{eq:admits}, up to lattice translations and the automorphisms of $\Lambda$ stabilising the Standard-Model embedding of S3. Representatives are
\begin{align}
\label{eq:reps}
V_2^{\rm A}&=\bigl(\tfrac23,\tfrac23,-\tfrac13,-\tfrac13,-\tfrac13,0,0,0\,\bigl|\,\tfrac12,\tfrac12,0^6\bigr),\nonumber\\
V_2^{\rm B}&=\bigl(\tfrac23,\tfrac23,-\tfrac13,-\tfrac13,-\tfrac13,-\tfrac12,-\tfrac12,0\,\bigl|\,0^8\bigr).
\end{align}
\end{theorem}

The invariant $\mu$ of the proof needs only that no such automorphism exchanges the two $E_8$ factors, which S3 forces; every element of $W(E_8)\times W(E_8)$ then acts blockwise and isometrically on the second factor. The count is therefore insensitive to how narrowly the group is drawn: the enumeration of appendix~\ref{app:classes} merges classes under a set of generators of the group, so the true orbits number at most two, and $\mu$ forbids the merger of the two, so they number at least two.

\begin{proof}
At this budget the doubling is rigid. The two states differ by a vector of the gauge lattice which the budget forces to be a root, by proposition~\ref{prop:sun}(a) at $n=2$, since $4R=94/45<4$; the root is orthogonal to the Standard-Model Cartan directions, so by part (d) the pair is a doublet of an $SU(2)$ of the local gauge group commuting with the Standard Model. Each sector named in S1 is constructed by one translation-free element, so the local shift is $V_2$ at $T(0,1)$ and $5V_2\equiv-V_2$ modulo $\Lambda$ at $T(0,5)$, with no Wilson-line contribution; the two cosets differ, because $V_2$ has order six and so $2V_2\notin\Lambda$. The conditions \eqref{eq:admits} are unchanged under $p\mapsto-p$, $\Delta\mapsto-\Delta$, since they admit each weight together with its negative, so a class admits a doubling at $T(0,5)$ exactly when it does at $T(0,1)$, and in either sector the class of the momentum fixes the class of $V_2$. Writing the momentum in the natural parametrisation splits the residual length into three squares whose integrality admits exactly eight values of the free parameter, and Cauchy--Schwarz applied to the pairing of the momentum with the root collapses those eight according to which $E_8$ factor carries the root: if it is the factor without the Standard Model, the second-factor part of the momentum is pinned to half a root and only two of the eight values survive; if it is the factor with the Standard Model, that part vanishes outright. What remains is a three-square problem in each branch. Its enumeration (appendix~\ref{app:classes}) gives $12{,}600$ solutions in $630$ classes of $\Lambda$, which fall into exactly two orbits; the conjugate weights double both counts without adding an orbit. The two are inequivalent because the invariant $\mu(V_2)=\min\{|y|^2:y\in p_B+\Lambda_{E_8}\}$, where $p_B$ is the second-factor component, takes the value $\tfrac12$ on the first class and $0$ on the second, and is preserved by every automorphism stabilising the Standard-Model embedding of S3, such automorphisms acting blockwise.
\end{proof}

In physical terms, the doubling root lies in the $E_8$ factor without the Standard Model in the first class and in the factor with it in the second, and no automorphism of the lattice that fixes the Standard Model exchanges the two factors. Other geometries do not raise the question: every other twisted sector of the seven geometries we examined, the two of section~\ref{sec:verif} and the five further symmetric orbifolds $\ZZ_4$, $\ZZ_6$-I, $\ZZ_6$-II, $\ZZ_3\times\ZZ_3$ and $\ZZ_4\times\ZZ_4$ of ref.~\cite{Fischer:2012qj}, each on the space group of its \textsc{orbifolder} geometry file with suffix \texttt{\_1\_1}, has $B\le13/16$, below the \ten{} threshold.

Both classes are realised. With $V_2$ fixed to one of \eqref{eq:reps}, an embedding is a pair of vectors $(V_1,W_3)$, $W_4=W_3$, with $3V_1,3W_3\in\Lambda$, without discrete torsion and with $t_Y$ as in S3, satisfying the complete set of modular-invariance conditions imposed by the \textsc{orbifolder}. We require every $SU(3)_c$ and $SU(2)_L$ root to pair integrally with $V_1$ and with $W_3$, each of the twelve roots $\pm(e_i-e_j)$, $i\le3<j\le5$, of the first-factor $SU(5)$ outside $SU(3)_c\times SU(2)_L$ to pair non-integrally with $V_1$ or with $W_3$, and $\Delta\cdot V_1$ and $\Delta\cdot W_3$ to be integers, which section~\ref{sec:signature} shows to be necessary for both members of the pair to survive. These conditions, and the four-dimensional gauge group, depend only on the classes $V_1+\Lambda$ and $W_3+\Lambda$; the spectrum does not. Replacing $V_1$ by $V_1+\lambda$ or $W_3$ by $W_3+\mu$, with $\lambda,\mu\in\Lambda$, changes the projection phases of twisted states by amounts that on this geometry depend only on $\lambda\cdot V_2$ and $\mu\cdot V_2$ modulo one, which modular invariance restricts to $0$, $\tfrac13$ and $\tfrac23$; this is the equivalence between lattice shifts and discrete torsion of ref.~\cite{Ploger:2007iq}. Each pair of classes therefore carries nine choices of these phases, and so up to nine spectra; for the model of section~\ref{sec:construction:model} the nine spectra are distinct, and each of the lattice shifts of $V_1$, $V_2$ and $W_3$ we tried there reproduces one of them. We identify pairs of classes related by the reflections in the roots of $\Lambda$ orthogonal to $W$, $\Delta$ and $V_2$, of which there are $138$ for the first class ($12$ in the first $E_8$ and $126$ in the second) and $242$ for the second ($2$ and $240$). This leaves $2{,}912$ pairs over the first class and $1{,}432$ over the second, inequivalent under these reflections, though possibly not under every automorphism of the data. In $4{,}032$ of them $SU(3)_c$ and $SU(2)_L$ are whole factors of the four-dimensional gauge group and $t_Y$ lies in the span of its abelian generators, orthogonal to all its simple roots; we say that these have the Standard Model intact. Over each of them we computed the spectra of nine embeddings, one for each choice of $\lambda\cdot V_2$ and $\mu\cdot V_2$ modulo one, $36{,}288$ in all. Of these, $21{,}168$ carry a doubled \ten{} representation at a twisted fixed point. None has three net generations in $Q$, $u^c$ and $e^c$ at once, the net number of a representation being the number of its states less the number of states in the conjugate representation, with the colour representation of each state read from its gauge momentum; the best reach two of the three, in $216$. Within S1--S4 and the conditions of section~\ref{sec:signature}, the search is exhaustive.

The two classes also account for the $R=47/90$ row of table~\ref{tab:ceiling}, where the largest multiplicity in the catalogues is one. Every model of the $\ZZ_3\times\ZZ_6$ catalogue but one\footnote{Its Standard-Model roots lie in both $E_8$ factors, and it carries no \ten{} multiplet at this budget.} embeds the Standard Model in an $SU(5)$ whose Cartan subalgebra is $W$, so theorem~\ref{thm:twoclasses} applies to it after an automorphism of $\Lambda$, and in none of them, evaluated in its own frame, does \eqref{eq:admits} hold at $T(0,1)$ or $T(0,5)$. Accordingly none of the $65$ \ten{} multiplets at $R=47/90$ has a root $\Delta\perp W$ with $|p\cdot\Delta|=1$, and none has a partner even before the orbifold projection. Multiplicity two at this budget is realised by the model of section~\ref{sec:construction}, which is built over the first admitting class, while the catalogue contains no model over either class. That no embedding over those classes with the Standard Model intact has three generations is consistent with the row, but the row is already explained by the shift classes rather than by the orbifold projection.

%%%%%%%%%%%%%%%%%%%%%%%%%%%%%%%%%%%%%%%%%%%%%%%%%%%%%%%%%%%%%%%%%%%%%%%%%%%%%%%
\section{Doubled representations at translation-free constructing elements}
\label{sec:signature}

On $T^6/(\ZZ_3\times\ZZ_6)$ with the lattice this geometry fixes, the sectors $T(0,1)$ and $T(0,5)$ each have a single constructing space-group element, and it carries no lattice translation: the twist rotates two of the three tori by $\pm60^\circ$, so each of them has exactly one fixed point, and the states of the sector are localised in those two tori and free to move in the third. The orbifold projection runs over all space-group elements that commute with the constructing element, and a Wilson line along a fixed torus acts as an additional projector \cite{Forste:2004ie}. Here the commuting elements include the whole point group and every lattice translation in the torus the twist leaves fixed, and on this geometry the only Wilson line the lattice permits to be non-trivial lies in that torus. The projection at this element therefore tests the difference of two surviving gauge momenta against $V_1$, $V_2$ \emph{and every Wilson line}, rather than against the local shift alone.

Comparing the projection phases of two states at such an element requires more than a shared constructing element: the states must also share their right-moving state, both momentum and excitation, and their left-moving oscillator excitation. At these two elements all three requirements follow. The right-moving masslessness condition $\tfrac12|q+v_g|^2+N_R=\tfrac12-\delta c(g)$, with $q$ a weight of $SO(8)$ in its vector or spinor class, $v_g$ the twist of the point-group part of $g$ and $N_R\ge0$ the right-moving oscillator number \cite{Nilles:2011aj}, involves neither the gauge momentum nor the left-moving oscillators. At $T(0,1)$, where $v_g=(0,\tfrac16,0,-\tfrac16)$, the first entry referring to the non-compact directions, and $\tfrac12-\delta c=\tfrac{13}{36}$, its only bosonic solutions are $q+v_g=(0,-\tfrac56,0,-\tfrac16)$ and $(0,\tfrac16,0,\tfrac56)$, one for each four-dimensional chirality and each with one fermionic partner; at $T(0,5)$ they are the negatives of these. Every left-chiral state at the element therefore carries the same right-moving momentum, and so do the $30{,}922$ multiplets at these two sectors in one embedding over each of the $2{,}912$ pairs of classes over the first shift class of section~\ref{sec:construction:classes} and in the $\ZZ_3\times\ZZ_6$ catalogue. The right-moving excitation is shared as well: in the vector and spinor classes alike $\tfrac12|q+v_g|^2\ge\tfrac{13}{36}$ at both elements, so no massless state there carries one. For two states distinct modulo $\Wsm$ the left-moving oscillator excitation follows from theorem~\ref{thm:bound}: a \ten{} state here has $R=47/90$ without an oscillator and $R=17/90$ with the smallest one, any larger excitation leaving no massless \ten{} state, and both $\sqrt{47/90}+\sqrt{17/90}$ and $2\sqrt{17/90}$ fall below $\sqrt2$, so both states of a doubled \ten{} representation are oscillator-free. In both populations the multiplets whose momenta have the length of the doubled states of table~\ref{tab:doubled}, $2{,}239$ and $12{,}350$ of them, carry no oscillator, while the nine same-representation pairs at these translation-free elements all differ in their oscillator excitations, and for none of them does the conclusion below hold.

If the two states share their right-moving state and their left-moving oscillator excitation, every term of the projection phase except the one linear in the gauge momentum cancels between them, and their difference pairs integrally with every gauge shift of the model. It is then a root of the unbroken four-dimensional gauge group, orthogonal to the Standard-Model Cartan subalgebra, and the two states are the two weights of an unbroken $SU(2)$ commuting with the Standard Model. In this geometry and these two sectors, a doubled \ten{} representation whose two states are distinct modulo $\Wsm$ therefore appears in the spectrum as a single multiplet of hidden dimension two, a doublet of that $SU(2)$, and never as two multiplets. The two populations show the undoubled and the doubled case. At these two elements the \ten{} multiplets of the $\ZZ_3\times\ZZ_6$ catalogue, $76$ in all, have hidden dimension one, and those of the embeddings over the first shift class, one over each of its $2{,}912$ pairs of classes, $2{,}239$ in $1{,}812$ of them including that of section~\ref{sec:construction}, have hidden dimension two; in neither population does a \ten{} representation appear there twice.

The step does not generalise. At $527$ of the $533$ twisted constructing elements of the seven geometries of section~\ref{sec:construction:classes}, the elements commuting with the constructing element impose only part of the gauge-shift and Wilson-line conditions. On $T^6/(\ZZ_2\times\ZZ_4)$, whose Wilson lines occupy all three tori, this happens at every element, and the conditions that drop out are those of the Wilson lines of the two rotated tori. Where a constructing element carries a lattice translation, even the point-group shifts drop out, and no statement of this form can be made, including at the translated elements of $T(0,2)$, $T(0,3)$ and $T(0,4)$ of this same geometry. Across both catalogues, none of the roots of $\Lambda$ by which the $33{,}743$ same-representation pairs listed as two multiplets differ is a root of the four-dimensional group; for the pairs of equal oscillator number it is a root of the local group (section~\ref{sec:verif}). On another geometry it is the argument that carries over, not its conclusion, and the argument needs a translation-free constructing element and every non-trivial Wilson line inside the torus the sector's twist leaves fixed.

%%%%%%%%%%%%%%%%%%%%%%%%%%%%%%%%%%%%%%%%%%%%%%%%%%%%%%%%%%%%%%%%%%%%%%%%%%%%%%%
\section{Discussion}
\label{sec:discussion}

The bound of theorem~\ref{thm:bound} rests on the evenness of the gauge lattice and on the hypotheses H1--H5, and the projected norms of section~\ref{sec:norms} make it a condition on the sector, with thresholds $17/20$ for the \ten{} and $13/20$ for the \fbar{} representations. For several oscillator-free states with $R<1$ it says more: they form the fundamental representation of a local $SU(n)$ commuting with the Standard Model, and the multiplicity ceiling is the level-one weight of that representation. The \ten{} threshold is crossed by the explicit model of section~\ref{sec:construction}, and for the embedding S3 exactly two shift classes on its geometry allow the doubling. Theorem~\ref{thm:bound}, proposition~\ref{prop:sun} and corollary~\ref{cor:ceiling} hold on any toroidal orbifold of the kind set up in section~\ref{sec:setup} whenever their hypotheses are met; theorem~\ref{thm:twoclasses} is specific to the geometry and embedding of S1 and S3, and our statements about spectra concern only the two geometries whose catalogues we computed.

For model building, the thresholds test a sector from its twist alone. No twisted sector of $\ZZ_2\times\ZZ_4$ reaches $17/20$, so no \ten{} representation repeats at a fixed point there, and the same budget arithmetic places every twisted sector of the five further geometries of section~\ref{sec:construction:classes} below that threshold, a prediction that only a computation of their spectra can test. Where a repetition is allowed, such states carry, besides the weights of the local $SU(n)$, a non-zero common charge $c_\perp$ unless $n=1/(1-R)$ exactly. In the catalogues the integer part of the ceiling is reached at four budgets. For two of the models, that of section~\ref{sec:construction} and a second embedding over the same shift class, an independent code for the orbifold projection also reproduces the oscillator-free twisted spectrum at every constructing element and right-mover, including the absence of states where the spectrum has none. The embeddings over the two admitting shift classes add an observation of more direct phenomenological interest: none of the $36{,}288$ with the Standard Model intact, and so none of the $21{,}168$ that realise the doubling, has three net generations of $Q$, $u^c$ and $e^c$ at once. This is a computational finding, and we have no derivation of it.

The case excluded by H5, two multiplets on one highest-weight momentum, is settled by the spectra rather than by the lattice, except for the \ten{} representations on $\ZZ_2\times\ZZ_4$ (section~\ref{sec:verif:h5}), and on another geometry with a small oscillator mode the spectra must be checked before the bound is used to exclude a repetition. Whether the local root carried by a repeated oscillator-free representation at $R<1$ lifts to a root of the four-dimensional group is open in general; section~\ref{sec:signature} identifies two sectors in which it does, and in the catalogues it never does for a pair listed as two multiplets. Theorem~\ref{thm:twoclasses} assumes the embedding S3, and an embedding in which $W$ sits differently inside $\Lambda$ requires the parametrisation and the root count to be redone.

%%%%%%%%%%%%%%%%%%%%%%%%%%%%%%%%%%%%%%%%%%%%%%%%%%%%%%%%%%%%%%%%%%%%%%%%%%%%%%%
\acknowledgments

The spectra of every model discussed here were computed with the publicly available \textsc{orbifolder} \cite{Nilles:2011aj}, and this work would not have been possible without it.

Large language models (Anthropic's Claude) were used in this work to generate code and drafts of text. All computations and results, and the text of this paper, have been verified by the author, who takes full responsibility for the content.

\section*{Data availability}
\addcontentsline{toc}{section}{Data availability}

The results of sections~\ref{sec:norms}--\ref{sec:ceiling} are proved in the text. The two classes of theorem~\ref{thm:twoclasses} and their inequivalence can be checked by hand from \eqref{eq:reps} and appendix~\ref{app:classes}; the completeness of the classification rests on the enumeration specified there, whose $12{,}600$ solutions in $630$ classes, closure computation over $332{,}640$ generator applications and breadth-first certificate, with an independent checker, are in the reproduction package. The model of section~\ref{sec:construction} is fully specified by \eqref{eq:model}, $W_4=W_3$, and the geometry file \texttt{Geometry\_Z3xZ6\_1\_1} of the \textsc{orbifolder} 1.2.1 distribution, which the package includes, and its spectrum can be regenerated with the \textsc{orbifolder} \cite{Nilles:2011aj}. The package also contains the per-model spectra of both catalogues, the assignment of Standard-Model representations and of $|w|^2$ by projection, the per-fixed-point multiplicity tables and the pairwise test behind tables~\ref{tab:verif}, \ref{tab:hist} and \ref{tab:ceiling}, with the tests of \eqref{eq:meanceiling} and of the rank restriction of section~\ref{sec:ceiling}, the input deck and stored spectrum of the model with the second code for the orbifold projection and its output at every constructing element and right-mover, the shift-class enumeration with its output and log, the spectra of all $36{,}288$ embeddings of section~\ref{sec:construction:classes} with the scripts that build and analyse them, the survey of the $533$ twisted constructing elements behind section~\ref{sec:signature}, and scripts that regenerate every number in table~\ref{tab:verif} and section~\ref{sec:construction}. Data are released under CC BY 4.0 and code under the MIT licence, with redistributed \textsc{orbifolder} input files under the GPL v3 terms they carry, and the package is permanently archived at Zenodo, DOI \href{https://doi.org/10.5281/zenodo.23575403}{\texttt{10.5281/zenodo.23575403}}.

\appendix

%%%%%%%%%%%%%%%%%%%%%%%%%%%%%%%%%%%%%%%%%%%%%%%%%%%%%%%%%%%%%%%%%%%%%%%%%%%%%%%
\section{The two shift classes}
\label{app:classes}

\subsection{Parametrisation}
\label{app:classes:param}

Under S1--S4 the momentum has $|p|^2=31/18$ and $|w|^2=6/5$, so $R=47/90$ and $4R=94/45<4$. By proposition~\ref{prop:sun} at $n=2$ the difference $\Delta$ of the two highest-weight momenta is a root, $|p\cdot\Delta|=1$, and the residual freedom is a vector $s'\perp\Delta$, $s'\perp W$ of the fixed norm $R-\tfrac12=1/45$. Under S3 the $W$-projection of $p$ is the traceless part of its first five coordinates, one of the weights $w_{ij}$ of \eqref{eq:wij}. Writing the first five coordinates of $p$ as $m(1,1,1,1,1)+w_{ij}$ and setting $B_v:=6(m-\tfrac25)$, the residual length splits as
\begin{equation}
\label{eq:split3}
R\;=\;5m^2+\bigl(p_6^2+p_7^2+p_8^2\bigr)+|p_B|^2\;=\;\tfrac{47}{90},
\end{equation}
where $p_B$ is the second-factor block and all three terms are non-negative. Since $p$ and $V_2$ differ by an element of $\Lambda$ and $6\Lambda\subset\Lambda$, the order condition $6V_2\in\Lambda$ is exactly $6p\in\Lambda$. The first block of $6p$ has the entry $B_v+6$ at the positions $i$ and $j$, the entry $B_v$ at the other three of the first five, and $c_k:=6p_k$ at $k=6,7,8$; its membership in $\Lambda_{E_8}$ says that $B_v$ and the $c_k$ are integers together or half-odd-integers together and that $5B_v+12+c_6+c_7+c_8\in2\ZZ$.

\subsection{The eight residual budgets}

From \eqref{eq:split3}, $5m^2\le\tfrac{47}{90}$ bounds $B_v$ to $[-4.3391,-0.4609]$, and intersecting with $\tfrac12\ZZ$ leaves exactly eight values:
\begin{center}
\begin{tabular}{rccccc}
\toprule
$B_v$ & $-4$ & $-\tfrac72$ & $-3$ & $-\tfrac52$ \\
${\rm rem}=\tfrac{47}{90}-5m^2$ & $\tfrac16$ & $\tfrac{17}{48}$ & $\tfrac{17}{36}$ & $\tfrac{25}{48}$ \\
\midrule
$B_v$ & $-2$ & $-\tfrac32$ & $-1$ & $-\tfrac12$\\
${\rm rem}$ & $\tfrac12$ & $\tfrac{59}{144}$ & $\tfrac14$ & $\tfrac1{48}$\\
\bottomrule
\end{tabular}
\end{center}
The neighbouring values fail on their own: $B_v=-\tfrac92$ gives $5m^2=\tfrac{49}{80}>\tfrac{47}{90}$, and $B_v=0$ gives $5m^2=\tfrac45>\tfrac{47}{90}$.

\subsection{The two branches}

The lattice $\Lambda$ is an orthogonal sum of two even lattices, so $2=|\Delta|^2=|\Delta_1|^2+|\Delta_2|^2$ with both terms even: one is $2$ and the other $0$, and there is no mixed case. With $\Delta\perp W$ this leaves $20$ candidates in the factor containing the Standard Model and all $240$ roots of the other, $260$ in all. If $\Delta$ lies in the factor without the Standard Model, then $|p_B\cdot\Delta|=1$ and Cauchy--Schwarz gives $|p_B|^2\ge\tfrac12$, while \eqref{eq:split3} gives $|p_B|^2\le\tfrac{47}{90}$; since $|6p_B|^2$ is an even integer in $[18,18.8]$ it is $18$, Cauchy--Schwarz is saturated, $6p_B=\pm3\Delta$ and $|p_B|^2=\tfrac12$ exactly, so ${\rm rem}\ge\tfrac12$ and only $B_v=-\tfrac52$ and $B_v=-2$ survive. If instead $\Delta$ lies in the factor containing the Standard Model, the same inequality applied to the part of the first block orthogonal to $W$ gives $|p_B|^2\le\tfrac{47}{90}-\tfrac12=\tfrac1{45}$, so $|6p_B|^2\le\tfrac45<2$ and $p_B=0$ exactly. The branches are disjoint, and each leaves a three-square problem $c_6^2+c_7^2+c_8^2=T$ with $T$ read off the table. Enumerating both branches over all eight budgets, all ten weights of the \ten{} family and all $260$ candidate roots gives $12{,}600$ solutions in $630$ classes of $\Lambda$, and adjoining the conjugate weights doubles both counts.

\subsection{Representatives and hand checks}

Table~\ref{tab:reps} lists the two classes of theorem~\ref{thm:twoclasses} in one frame. The first five coordinates of the representatives \eqref{eq:reps} are identical, carrying the weight $w_{12}$ at $B_v=-2$, and the doubling pair $\pm(\tfrac12,\tfrac12)$ sits in coordinates $6,7$ of the Standard Model's own $E_8$ in class B and in coordinates $1,2$ of the other $E_8$ in class A; that is the entire difference between them. Both $V_2^{\rm A}$ and $V_2^{\rm B}$ carry the weight $w_{12}$, which is a $u^c$, at $V_2\cdot t_Y=-\tfrac23$; the doubled fixed point of the model \eqref{eq:model} lies in class A at the representative $p_1$ of table~\ref{tab:doubled}, which carries $w_{35}$ and is a $Q$ at $p\cdot t_Y=+\tfrac16$, and $p_1-V_2^{\rm A}\in\Lambda$; which member of the \ten{} family is doubled depends on the representative, not on the class.

\begin{table}[h]
\centering
\small
\begin{tabular}{lll}
\toprule
class & $V_2$ & $\Delta$ (block 1 $\mid$ block 2)\\
\midrule
A & $V_2^{\rm A}$ of \eqref{eq:reps} & $\bigl(0^8\mid1,1,0^6\bigr)$\\
B & $V_2^{\rm B}$ of \eqref{eq:reps} & $\bigl(0,0,0,0,0,-1,-1,0\mid0^8\bigr)$\\
\bottomrule
\end{tabular}
\caption{The two admitting classes. In each row the two states are $V_2$ and $V_2-\Delta$.}
\label{tab:reps}
\end{table}

For either row: $|V_2|^2=2\cdot\tfrac49+3\cdot\tfrac19+2\cdot\tfrac14=\tfrac{31}{18}$; $|P_WV_2|^2=\tfrac{11}9-\tfrac15\cdot\tfrac19=\tfrac65$, the traceless part of $(\tfrac23,\tfrac23,-\tfrac13,-\tfrac13,-\tfrac13)$ being $(\tfrac35,\tfrac35,-\tfrac25,-\tfrac25,-\tfrac25)=w_{12}$; $|\Delta|^2=2$ with $\Delta$ vanishing on the first five coordinates, hence $\Delta\perp W$; $V_2\cdot\Delta=+1$; $6V_2\in\Lambda$ --- block 1 of $6V_2^{\rm B}$ is $(4,4,-2,-2,-2,-3,-3,0)$, integral with even sum --- while $3V_2\notin\Lambda$, so the order is exactly six; and $|V_2-\Delta|^2=|V_2|^2$ with $P_W(V_2-\Delta)=P_WV_2$, so the partner is a second state of the same representation at the same fixed point. The search imposes only $6V_2\in\Lambda$, which allows orders $1,2,3,6$; every solution nevertheless has order exactly six. Modular invariance imposes nothing at this stage: the condition $N(|V_2|^2-|v_2|^2)\in2\ZZ$ evaluates to $6(\tfrac{31}{18}-\tfrac1{18})=10$ for every solution.

\subsection{The separating invariant}

For a class $V_2+\Lambda$ let $\mu(V_2)=\min\{|y|^2:y\in p_B+\Lambda_{E_8}\}$. It is well defined on classes and preserved by every automorphism stabilising the Standard-Model embedding of S3, since those act blockwise --- an automorphism exchanging the two $E_8$ factors would carry the Standard Model out of the first --- and blockwise action is by isometries. Now $\mu(V_2^{\rm B})=0$ because $p_B=0$, and $\mu(V_2^{\rm A})=\tfrac12$: any $y$ in that class has $2y\in\Lambda_{E_8}$ with $2y\equiv\Delta$ modulo $2\Lambda_{E_8}$, and since the minimal norm of $2\Lambda_{E_8}$ is $8$ while $|\Delta|^2=2$, the vector $2y$ is non-zero; every non-zero vector of $\Lambda_{E_8}$ has squared length at least $2$, so $|2y|^2\ge2$ and $|y|^2\ge\tfrac12$, with equality at $y=\Delta/2$. The two classes are therefore inequivalent.

\subsection{Completeness of the orbit count}

The stabiliser is generated by $264$ elements: the reflections in the three Standard-Model simple roots, in the $20$ first-factor roots orthogonal to $W$ and in the $240$ roots of the second factor, together with the global $-1$. The $1{,}260$ classes of the enumeration with its conjugate branch adjoined form a set closed under these generators: all $332{,}640$ images lie in it. The connected components of the graph the generators induce on this set are therefore orbits, and there are two, of sizes $600$ and $30$ in the \ten{} branch. A breadth-first search from the two representatives reaches all $1{,}260$ classes and gives for each a chain of generators carrying it to its representative. The invariant $\mu$ shows that the two orbits do not merge under any element of the stabiliser, not only under the generators used.

%%%%%%%%%%%%%%%%%%%%%%%%%%%%%%%%%%%%%%%%%%%%%%%%%%%%%%%%%%%%%%%%%%%%%%%%%%%%%%%
\section{Counting conventions}
\label{app:defs}

A \emph{model} is one entry of a published catalogue: a lattice, two shift vectors and six Wilson lines. We use one space group of each geometry. The published $\ZZ_2\times\ZZ_4$ set runs to $55{,}429$ models across eleven space groups and the $\ZZ_3\times\ZZ_6$ set to $5{,}153$ across two \cite{Parr:2020oar}, and the space groups we use hold $11{,}540$ and $4{,}612$ distinct models. Ten models of the second geometry were computed twice and are counted once.

A \emph{multiplet} is one irreducible representation of the four-dimensional gauge group, with its multiplicity, its sector, fixed point, gauge momentum, right-mover and oscillators. We count left-chiral multiplets only, never their CPT partners separately. A multiplet carries a Standard-Model representation if its colour, isospin and hypercharge say so, whatever it carries under the remaining gauge factors. An \emph{embedding} is the shift and Wilson-line data together with a choice of hypercharge generator. The massless spectrum does not depend on that choice, which fixes only which states carry which Standard-Model representation.

A constructing element $g=(\theta^k\omega^l,\sum_an_ae_a)$ is labelled by the full tuple $(k,l,n_1\dots n_6)$, the finer of the two available orbit notions, so distinct constructing elements are never merged. Across both catalogues only $128$ twisted constructing elements occur, $58$ on the first geometry and $70$ on the second, and the $880{,}589$ populated fixed points of table~\ref{tab:verif} count each of these $128$ elements once per model in which it carries massless states. Multiplicities and the ceiling refer to a Standard-Model representation at a populated fixed point: the $109{,}717$ \ten{} representations sit at $45{,}450$ populated fixed points, $2.414$ per fixed point, and the $360{,}070$ \fbar{} ones at $198{,}385$, $1.815$ per fixed point. A \emph{same-representation pair} is an unordered pair of distinct left-chiral multiplets sharing a populated fixed point and a Standard-Model representation, whatever their hidden representations.

For each model we group the left-chiral multiplets at twisted fixed points by constructing element, assign Standard-Model representations and $|w|^2$ by projection onto $W$, form multiplicities, and test every pair within a group against \eqref{eq:mainbound}. Each model is counted once, at the first listed embedding for which the Standard Model can be identified, with one convention on both geometries; where several placements of the Standard-Model factors among the model's gauge factors reproduce its listed table, we take the first. On the first geometry this succeeds for $11{,}526$ of the $11{,}540$ models, $10{,}927$ at the first embedding and $599$ at a later one; the remaining $14$ list a single embedding, at $k_Y=5/3$, and are not counted. On the second geometry it succeeds for all $4{,}612$, so every count of table~\ref{tab:verif} covers the same $16{,}138$ models. The \textsc{orbifolder} lists multiplicities only for the embedding it selects, so we computed the other identifications separately on both catalogues; none of them yields a violation of \eqref{eq:mainbound}, a repeated \ten{} representation, or a \ten{} multiplet of hidden dimension above one.


\begin{thebibliography}{99}

\bibitem{Gross:1984dd}
D.~J.~Gross, J.~A.~Harvey, E.~J.~Martinec and R.~Rohm,
\emph{The Heterotic String},
Phys.\ Rev.\ Lett.\ \textbf{54} (1985) 502.

\bibitem{Gross:1985fr}
D.~J.~Gross, J.~A.~Harvey, E.~J.~Martinec and R.~Rohm,
\emph{Heterotic String Theory. 1. The Free Heterotic String},
Nucl.\ Phys.\ B \textbf{256} (1985) 253.

\bibitem{Dixon:1985jw}
L.~J.~Dixon, J.~A.~Harvey, C.~Vafa and E.~Witten,
\emph{Strings on Orbifolds},
Nucl.\ Phys.\ B \textbf{261} (1985) 678.

\bibitem{Dixon:1986jc}
L.~J.~Dixon, J.~A.~Harvey, C.~Vafa and E.~Witten,
\emph{Strings on Orbifolds. 2.},
Nucl.\ Phys.\ B \textbf{274} (1986) 285.

\bibitem{Ibanez:1986tp}
L.~E.~Ib\'a\~nez, H.~P.~Nilles and F.~Quevedo,
\emph{Orbifolds and Wilson Lines},
Phys.\ Lett.\ B \textbf{187} (1987) 25.

\bibitem{Hamidi:1986vh}
S.~Hamidi and C.~Vafa,
\emph{Interactions on Orbifolds},
Nucl.\ Phys.\ B \textbf{279} (1987) 465.

\bibitem{Dixon:1986qv}
L.~J.~Dixon, D.~Friedan, E.~J.~Martinec and S.~H.~Shenker,
\emph{The Conformal Field Theory of Orbifolds},
Nucl.\ Phys.\ B \textbf{282} (1987) 13.

\bibitem{Forste:2004ie}
S.~F\"orste, H.~P.~Nilles, P.~K.~S.~Vaudrevange and A.~Wingerter,
\emph{Heterotic brane world},
Phys.\ Rev.\ D \textbf{70} (2004) 106008
[\href{https://arxiv.org/abs/hep-th/0406208}{hep-th/0406208}].

\bibitem{Kobayashi:2004ya}
T.~Kobayashi, S.~Raby and R.-J.~Zhang,
\emph{Searching for realistic 4d string models with a Pati-Salam symmetry: Orbifold grand unified theories from heterotic string compactification on a $Z_6$ orbifold},
Nucl.\ Phys.\ B \textbf{704} (2005) 3
[\href{https://arxiv.org/abs/hep-ph/0409098}{hep-ph/0409098}].

\bibitem{Buchmuller:2005jr}
W.~Buchm\"uller, K.~Hamaguchi, O.~Lebedev and M.~Ratz,
\emph{Supersymmetric standard model from the heterotic string},
Phys.\ Rev.\ Lett.\ \textbf{96} (2006) 121602
[\href{https://arxiv.org/abs/hep-ph/0511035}{hep-ph/0511035}].

\bibitem{Buchmuller:2006ik}
W.~Buchm\"uller, K.~Hamaguchi, O.~Lebedev and M.~Ratz,
\emph{Supersymmetric Standard Model from the Heterotic String (II)},
Nucl.\ Phys.\ B \textbf{785} (2007) 149
[\href{https://arxiv.org/abs/hep-th/0606187}{hep-th/0606187}].

\bibitem{Lebedev:2006kn}
O.~Lebedev, H.~P.~Nilles, S.~Raby, S.~Ramos-S\'anchez, M.~Ratz, P.~K.~S.~Vaudrevange and A.~Wingerter,
\emph{A Mini-landscape of exact MSSM spectra in heterotic orbifolds},
Phys.\ Lett.\ B \textbf{645} (2007) 88
[\href{https://arxiv.org/abs/hep-th/0611095}{hep-th/0611095}].

\bibitem{Kobayashi:2006wq}
T.~Kobayashi, H.~P.~Nilles, F.~Pl\"oger, S.~Raby and M.~Ratz,
\emph{Stringy origin of non-Abelian discrete flavor symmetries},
Nucl.\ Phys.\ B \textbf{768} (2007) 135
[\href{https://arxiv.org/abs/hep-ph/0611020}{hep-ph/0611020}].

\bibitem{Ploger:2007iq}
F.~Pl\"oger, S.~Ramos-S\'anchez, M.~Ratz and P.~K.~S.~Vaudrevange,
\emph{Mirage Torsion},
JHEP \textbf{04} (2007) 063
[\href{https://arxiv.org/abs/hep-th/0702176}{hep-th/0702176}].

\bibitem{Nilles:2008gq}
H.~P.~Nilles, S.~Ramos-S\'anchez, M.~Ratz and P.~K.~S.~Vaudrevange,
\emph{From strings to the MSSM},
Eur.\ Phys.\ J.\ C \textbf{59} (2009) 249
[\href{https://arxiv.org/abs/0806.3905}{arXiv:0806.3905}].

\bibitem{Nilles:2011aj}
H.~P.~Nilles, S.~Ramos-S\'anchez, P.~K.~S.~Vaudrevange and A.~Wingerter,
\emph{The Orbifolder: A Tool to study the Low Energy Effective Theory of Heterotic Orbifolds},
Comput.\ Phys.\ Commun.\ \textbf{183} (2012) 1363
[\href{https://arxiv.org/abs/1110.5229}{arXiv:1110.5229}].

\bibitem{Fischer:2012qj}
M.~Fischer, M.~Ratz, J.~Torrado and P.~K.~S.~Vaudrevange,
\emph{Classification of symmetric toroidal orbifolds},
JHEP \textbf{01} (2013) 084
[\href{https://arxiv.org/abs/1209.3906}{arXiv:1209.3906}].

\bibitem{CaboBizet:2013gns}
N.~G.~Cabo~Bizet, T.~Kobayashi, D.~K.~Mayorga~Pe\~na, S.~L.~Parameswaran, M.~Schmitz and I.~Zavala,
\emph{R-charge Conservation and More in Factorizable and Non-Factorizable Orbifolds},
JHEP \textbf{05} (2013) 076
[\href{https://arxiv.org/abs/1301.2322}{arXiv:1301.2322}].

\bibitem{Parr:2020oar}
E.~Parr, P.~K.~S.~Vaudrevange and M.~Wimmer,
\emph{Predicting the orbifold origin of the MSSM},
Fortsch.\ Phys.\ \textbf{68} (2020) 2000032
[\href{https://arxiv.org/abs/2003.01732}{arXiv:2003.01732}].

\end{thebibliography}
\end{document}